%% file: randommodel.tex
\documentclass{sig-alternate}
\conferenceinfo{KDD'10,} {July 25--28, 2010, Washington, DC, USA.}  
\CopyrightYear{2010} 
\crdata{978-1-4503-0055-1/10/07} 
\usepackage[ruled, vlined, linesnumbered, nofillcomment]{algorithm2e}

\usepackage[english]{babel}
\usepackage{amssymb}
\usepackage{amsmath}
\usepackage{subfigure}
\usepackage{booktabs}
\usepackage{cite}
\usepackage{url}
\usepackage{graphicx}

\input{defines}

\begin{document}

\title{Probably the Best Itemsets}
\numberofauthors{1}
\author{
\alignauthor
Nikolaj Tatti\\
	\affaddr{ADReM, University of Antwerp, Antwerpen, Belgium} \\
	\email{nikolaj.tatti@ua.ac.be}
}

\maketitle
\begin{abstract}
One of the main current challenges in itemset mining is to discover
a small set of high-quality itemsets.
In this paper we propose a new and general approach for measuring the quality
of itemsets. The method is solidly founded in Bayesian statistics and
decreases monotonically, allowing for efficient discovery of all interesting
itemsets. The measure is defined by connecting statistical models and
collections of itemsets. This allows us to score individual itemsets with the
probability of them occuring in random models built on the data.

As a concrete example of this framework we use exponential models. This class
of models possesses many desirable properties. Most importantly, Occam's razor
in Bayesian model selection provides a defence for the pattern explosion. As
general exponential models are infeasible in practice, we use decomposable
models; a large sub-class for which the measure is solvable. For the actual
computation of the score we sample models from the posterior distribution using
an MCMC approach.

Experimentation on our method demonstrates the measure works in practice and
results in interpretable and insightful itemsets for both synthetic and
real-world data.
\end{abstract}

\category{H.2.8}{Database management}{Database applications---Data mining}
\category{G.3}{Probability and Statistics}{Markov processes}

\terms{Algorithms, Theory}

\keywords{Itemset mining, exponential models, decomposable models, junction trees, Bayesian model selection, MCMC}

\vfill\eject 
\input{introduction}
\input{connect}
\input{model}
\input{algorithm}
\input{related}
\input{experiments}
\input{conclusions}
\section{Acknowledgments}
Nikolaj Tatti is funded by FWO postdoctoral mandate.

\bibliography{bibliography}
\input{appendix}

\end{document}

%% file: defines.tex
\usepackage{ifthen}
\newcommand{\set}[1]{\left\{#1\right\}}
\newcommand{\pr}[1]{\left(#1\right)}
\newcommand{\fpr}[1]{\mathopen{}\left(#1\right)}

\newcommand{\abs}[1]{{\left|#1\right|}}

\newcommand{\enset}[2]{\left\{#1 ,\ldots , #2\right\}}
\newcommand{\enpr}[2]{\pr{#1 ,\ldots , #2}}

\newcommand{\real}{\mathbb{R}}

\usepackage{tikz}
\usetikzlibrary{snakes,arrows,shapes}

\newcommand{\define}{\leftarrow}

\newcommand{\dispfunc}[2]{%
  \ensuremath{%
  \ifthenelse{\equal{\noexpand#2}{}}%
    {{#1}}%
    {{#1}\fpr{#2}}}}

\newcommand{\score}[1]{sc\fpr{#1}}

\newcommand{\np}{\textbf{NP}}

\newcommand{\link}[1]{\mathit{link}\fpr{#1}}
\newcommand{\tb}[1]{\mathit{st}\fpr{#1}}

\newcommand{\kl}[1]{\mathit{KL}\fpr{#1}}
\newcommand{\fm}[1]{\mathit{fam}\fpr{#1}}

\newcommand{\dgr}[1]{d\fpr{#1}}
\newcommand{\sdgr}[1]{\mathit{sd}\fpr{#1}}
\newcommand{\mdgr}[1]{\mathit{md}\fpr{#1}}
\newcommand{\rf}[1]{\mathit{rf}\fpr{#1}}
\newcommand{\gain}[1]{\mathit{gain}\fpr{#1}}

\newcommand{\mods}[1]{\mathcal{#1}}

\newcommand{\splt}[1]{\textsc{Split}\fpr{#1}}
\newcommand{\merge}[1]{\textsc{Merge}\fpr{#1}}

\newcommand{\ifam}[1]{\mathcal{#1}}
\newcommand{\tree}[1]{\mathcal{#1}}

\newcommand{\freq}[1]{\mathit{fr}\fpr{#1}}

\newcommand{\ent}[1]{H\fpr{#1}}

\newtheorem{theorem}{Theorem}
\newtheorem{lemma}[theorem]{Lemma}
\newtheorem{proposition}[theorem]{Proposition}
\newtheorem{corollary}[theorem]{Corollary}

\newtheorem{example}[theorem]{Example}

\SetKwComment{tcpas}{\{}{\}}
\SetKw{AND}{and}
\SetKw{OR}{or}
\SetKw{NONE}{none}
\SetAlgoSkip{}
\SetArgSty{textnormal}

%% file: introduction.tex
\section{Introduction}
Discovering frequent itemsets is one of the most active fields in data mining.
As a measure of quality, frequency possesses a lot of positive properties: it is
easy to interpret and as it decreases monotonically there exist 
efficient algorithms for discovering large collections of frequent
itemsets~\cite{agrawal:96:fast}.  However, frequency also has
serious drawbacks. A frequent itemset may be uninteresting if its elevated
frequency is caused by frequent singletons.  On the
other hand, some non-frequent itemsets could be interesting.
Another drawback is the problem of pattern explosion when mining with a
low threshold.

Many different quality measures have been suggested to overcome the mentioned
problems (see Section~\ref{sec:related} for a more detailed discussion). Usually these
measures compare the observed frequency to some expected value derived,
for example, from the independence model. Using such measures we may obtain better
results. However, these approaches still suffer from pattern explosion. To
point out the problem, assume that two items, say $a$ and $b$ are correlated,
and hence are considered significant. Then any itemset containing $a$ and $b$
will be also considered significant.

\begin{example}
\label{ex:problem}
Assume a dataset with $K$ items $a_1, \ldots, a_K$ such that $a_1$ and $a_2$
always yield identical value and the rest of the items are independently
distributed. Assume that we apply some statistical method to evaluate the
significance of itemsets by using the independence model as the ground truth.  If an
itemset $X$ contains $a_1a_2$, its frequency will be higher than then the
estimate of the independence model. Hence, given enough data, the P-value of
the statistical test will go to $0$, and we will conclude that the itemset $X$
is interesting. Consequently we will find $2^{K - 2}$ interesting itemsets.
\end{example}

In this work we approach the problem of defining quality measure from a
novel point of view. We construct a connection between itemsets and
statistical models and use this connection to define a new quality measure for
itemsets. To motivate this approach further, let us consider the following
example.

\begin{example}
\label{ex:toy}
Consider a binary dataset with $5$ items, say $a_1, \ldots, a_5$, generated
from the independence model. We argue that if we know that the data comes from
the independence model, then the only interesting itemsets are the singletons.
The reasoning behind this claim is that the frequencies of singletons
correspond exactly to the column margins, the parameters of the independence
model.  Once we know the singleton frequencies, there is nothing left in the
data that would be statistically interesting.

Let us consider a more complicated example. Assume that data is generated from
a Chow-Liu tree model~\cite{chow:68:approximating}, say 
\[
p(A) = p(a_1)p(a_2 \mid a_1)p(a_3 \mid a_1) p(a_4 \mid a_2) p(a_5 \mid a_4).
\]
Again, if we know that
data is generated from this model, then we argue that the interesting itemsets
are $a_1$, $a_2$, $a_3$, $a_4$, $a_5$, $a_1a_2$, $a_1a_3$, $a_2a_4$, and
$a_4a_5$. The reasoning here is the same as with the independence model. If we
know the frequencies of these itemsets we can derive the parameters of the
distribution. For example, $p(a_2 = 1 \mid a_1 = 1) = \freq{a_1a_2} /
\freq{a_1}$.
\end{example}

Let us now demonstrate that this approach will produce much smaller and
more meaningful output than the method given in Example~\ref{ex:problem}.

\begin{example}
Consider the data given in Example~\ref{ex:problem}. To fully describe the data
we only need to know the frequencies of the singletons and the fact that $a_1$ and
$a_2$ are identical. This information can be expressed by outputting the
frequencies of singleton itemsets and the frequency of itemset $a_1a_2$.
This will give us $K + 1$ interesting patterns in total.
\end{example}

Our approach is to extend the idea pitched in the preceding example to a
general itemset mining framework. In the example we knew which model generated
the data, in practice, we typically do not. To solve this we will use the Bayesian
approach, and instead of considering just one specific model, we will consider a
large collection of models, namely exponential models.  A virtue of these
models is that we can naturally connect each model to certain itemsets.  A
model $M$ has a posterior probability $P(M \mid D)$, that is, how probable is
the model given the data.  The score of a single itemset then is just the
probability of it being a parameter of a random model given the data. This
setup fits perfectly the given example. If we have strong evidence that data is
coming from the independence model, say $M$, then the posterior probability
$P(M \mid D)$ will be close to $1$, and the posterior probability of any other
model will be close to $0$. Since the independence model is connected to the
singletons, the score for singletons will be $1$ and the score for any other
itemset will be close to $0$.

Interestingly, using statistical models for defining significant itemsets
provides an approach to the problem of the pattern set explosion (see
Section~\ref{sec:compute} for more technical details). Bayesian model selection
has an in-built Occam's razor, favoring simple models over complex ones.  Our
connection between models and itemsets is such that simple models will
correspond to small collections of itemsets.  In result, only a small collection
of itemsets will be considered interesting, unless the data provides sufficient
amount of evidence.

Our contribution in the paper is two-fold. First, we introduce a general framework of
using statistical models for scoring itemsets in Section~\ref{sec:connect}.
Secondly, we provide an example of this framework in Section~\ref{sec:model} by
using exponential models and provide solid theoretical evidence that our
choices are well-founded. We provide the sampling algorithm in Section~\ref{sec:algorithm}.
We discuss related work in
Section~\ref{sec:related} and present our experiments in
Section~\ref{sec:experiments}. Finally, we conclude our work with
Section~\ref{sec:conclusions}. The proofs are given in Appendix. The implementation is
provided for research purposes\footnote{\url{http://adrem.ua.ac.be/implementations}}.

%% file: connect.tex
\section{Significance of Itemsets by Statistical Models}
\label{sec:connect}
As we discussed in the introduction, our goal is to define a quality
measure for itemsets using statistical models. In this section we provide a
general framework for such a score. We will define the actual models in the
next section.

We begin with some preliminary definitions and notations. In our setup a
\emph{binary dataset} is a collection of $N$ transactions, binary vectors of length $K$.  We
assume that these vectors are independently generated from some unknown
distribution.  Such a dataset can be easily represented by a binary matrix of size
$N \times K$.  By an attribute $a_i$ we mean a Bernoulli random variable
corresponding to the $i$th column of the data.  We denote the set of
\emph{attributes} by $A = \enset{a_1}{a_K}$.

An itemset $X$ is simply a subset of $A$. Given an itemset $X =
\enset{a_{i_1}}{a_{i_L}}$ and a transaction $t$ we denote by $t_X =
\enpr{t_{i_1}}{t_{i_L}}$ the projection of $t$ into $X$. We say that $t$ \emph{covers}
$X$ if all elements in $t_X$ are equal to 1.

We say that a collection of itemsets
$\ifam{F}$ is \emph{downward closed} if for each member $X \in \ifam{F}$ any
sub-itemset is also included. This property plays a crucial point in mining
frequent patterns since it allows effective candidate pruning in level-wise
approach and branch pruning in a DFS approach.

Our next step is to define the correspondence between statistical models and
families of itemsets. Assume that we have a set $\mods{M}$ of statistical
models for the data. We will discuss in later sections what specific models we
are interested in, but for the moment, we will keep our discussion on a high level.
Each model $M \in \mods{M}$ has a posterior probability $P(M \mid D)$, that is,
how probable the model $M$ is given data $D$. To link the models to
families of itemsets we assume that we have a function $\mathit{fam}$ that identifies a
model $M$ with a \emph{downward closed} family of itemsets. As we will see later on,
there is one particular natural choice for such a function. 

Now that we have our models that are connected to certain families of
itemsets, we are ready to define a score for individual itemsets. The score for an itemset
$X$ is the posterior probability of $X$ being a member in a family or itemsets,
\begin{equation}
\label{eq:defscore}
	\score{X} = \sum_{X \in \fm{M}, \atop M \in \mods{M}} P(M \mid D).
\end{equation}

The motivation for such score is as follows. If we are sure that some
particular model $M$ is the correct model for $D$, then the posterior
probability for that model will be close to $1$ and the posterior probabilities
for other models will be close to $0$. Consequently, the score for an itemset
$X$ will be close to $1$ if $X \in \fm{M}$, and $0$ otherwise.

Naturally, the pivotal choice of this score lies in the mapping
$\mathit{fam}$. Such a mapping needs to be statistically well-founded, and especially the
size of an itemset family should be reflected in the complexity of the
corresponding model.  We will see in the following section that a particular
choice for the model family and mapping $\mathit{fam}$ has these properties, and leads
to certain important qualities.

\begin{proposition}
The score decreases monotonically, that is, $X \subseteq Y$ implies $\score{X} \geq \score{Y}$.
\end{proposition}

\begin{proof}
We are allowing $\mathit{fam}$ only to map on downward closed families. Hence the inequality
\[
	\sum_{X \in \fm{M}, \atop M \in \mods{M}} P(M \mid D) \geq \sum_{Y \in \fm{M}, \atop M \in \mods{M}} P(M \mid D)
\]
holds, and consequently we have $\score{X} \geq \score{Y}$.
This completes the proof.
\end{proof}

%% file: model.tex
\section{Exponential Models}
\label{sec:model}
In this section we will make our framework more concrete by providing a
specific set $\mods{M}$ of statistical models and the function $\mathit{fam}$
identifying the models with families of itemsets. We will first give the
definition of the models and the mapping. After this, we point out the main
properties of our model and justify our choices. However, as, it turns out that
computing the score for these models is infeasible, so instead, we solve this
problem by considering decomposable models.

\subsection{Definition of the models}

Models of exponential form have been studied exhaustively in statistics, and
have been shown to have good theoretical and practical properties.  In our
case, using exponential models provide a natural way of describing the
dependencies between the variables. In fact, the exponential model class
contains many natural models such as, the independence model, the Chow-Liu tree model,
and the discrete Gaussian model. Finally, such models have been used successfully for
predicting itemset frequencies~\cite{pavlov:03:beyond} and ranking
itemsets~\cite{tatti:08:maximum}.

In order to define our model, let $\ifam{F}$ be a downward closed family of
itemsets containing all singletons. For an itemset $X \in \ifam{F}$ we define
an indicator function $S_X(t)$ mapping a transaction $t$ into binary value. If
the transaction $t$ covers $X$, then $S_X(t) = 1$, and $0$ otherwise. We define
an exponential model $M$ associated with $\ifam{F}$ to be the collection of 
distributions having the exponential form

\[
	P(A = t \mid M, r) = \exp\fpr{\sum_{X \in \ifam{F}} r_X S_X(t)},
\]
where $r_X$ is a parameter, a real value, for an itemset $X$. Model $M$ also
contains all the distributions that can be obtained as a limit of the
distribution having the exponential form. This technicality is needed to handle
distributions with zero probabilities.  Since the indicator function
$S_\emptyset(t)$ is equal to $1$ for any $t$, the parameter
$r_\emptyset$ acts like a normalization constant. The rest of the parameters
form a parameter vector $r$ of length $\abs{\ifam{F}} - 1$.  Naturally, we set
$\ifam{F} = \fm{M}$. 

\begin{example}
Assume that $\ifam{F}$ consists only of singleton itemsets. Then the corresponding
model has the form
\begin{equation}
\label{eq:expform}
	\exp\fpr{r_\emptyset + \sum_{i = 1}^K r_{a_i} S_{a_i}(t)} = \exp\fpr{r_\emptyset}\prod_i^K \exp\fpr{r_{a_i}S_{a_i}(t)}.
\end{equation}
Since $S_{a_i}(t)$ depends only on $t_i$, the model is actually the independence model.
The other extreme is when $\ifam{F}$ consists of all itemsets. Then we can
show that the corresponding model contains all possible distributions, that is,
the model is in fact the parameter-free model.

As an intermediate example, the tree model in Example~\ref{ex:toy} is also an
exponential model with a corresponding family $\enset{a_1}{a_5} \cup
\set{a_1a_2, a_2a_3, a_2a_4, a_4a_5}$.

The intuition behind the model is that when an itemset $X$ is an element of $\ifam{F}$, then
the dependencies between the items in $X$ are considered important in the
corresponding model. For example, if $\ifam{F}$ consists only of singletons,
then there are no important correlations, hence the corresponding model should
be the independence model. On the other hand, in the tree model given in
Example~\ref{ex:toy} the important correlations are the parent-child item pairs,
namely, $a_1a_2, a_2a_3, a_2a_4, a_4a_5$. These are exactly, the itemsets (along with
the singletons) that correspond to the model.
\end{example}

Our choice for the models is particularly good since the complexity of models
reflects the size of itemset family. Since Bayesian approach has an in-built
tendency to punish complex families (we will see this in
Section~\ref{sec:compute}), we are punishing large families of itemsets.  If
the data states that simple models are sufficient, then the probability of
complex models will be low, and consequently the score for large itemsets will
also be low. In other words, we casted the problem of pattern set explosion
into a model overfitting problem and used Occam's razor to punish the
complex models!

\subsection{Computing the Model}
\label{sec:compute}
Now that we have defined our model $M$, our next step is to compute the
posterior probability $P(M \mid D)$. That is, the probability of $M$ given the
data set $D$. We select the model prior $P(M)$ to be uniform.  Recall that in
Eq.~\ref{eq:expform} a model $M$ has a set of parameters, that is, to
pinpoint a single distribution in $M$ we need a set of parameters $r$.
Following the Bayesian approach to compute $P(M \mid D)$ we need to marginalize
out the nuisance parameters $r$,
\[
	P(M \mid D) = \int_r P(M, r \mid D) \propto \int_r P(r \mid M) \prod_{t \in D} P(t \mid M, r).
\]
In the general case, this integral is too complex to solve analytically so we employ
the popular BIC estimate~\cite{schwarz:78:estimating},
\begin{equation}
\label{eq:bic}
	P(M \mid D) \approx C \times P(D \mid M, r^*)  \times \exp\fpr{-\log{\abs{D}}\frac{\abs{\ifam{F}} - 1}{2}},
\end{equation}
where $C$ is a constant and $r^*$ is the maximum likelihood estimate of the
model parameters. This estimate is correct when
$\abs{D}$ approaches infinity~\cite{schwarz:78:estimating}. So instead of
computing a complex integral our challenge is to discover the maximum
likelihood estimate $r^*$ and compute the likelihood of the data.
Unfortunately, 
using such model is an \np-hard problem (see, for example,~\cite{tatti:06:computational}). We will remedy this problem
in Section~\ref{sec:decompose} by considering a large subclass of exponential models
for which the maximum likelihood can be easily computed.

\subsection{Justifications for Exponential Model}
In this section we will provide strong theoretical justification for our choices
and show that our score fulfills the goals we set in the introduction.

We saw in Example~\ref{ex:toy} that if the data comes from the independence
model, then we only need the frequencies of the singleton itemsets to
completely explain the underlying model.  The next theorem shows that this
holds in general case.

\begin{theorem}
\label{thr:itemsetsderive}
Assume that data $D$ is generated from a distribution $p$ that comes from an
exponential model $M$. Let $\ifam{F} = \fm{M}$ be the family of
itemsets.  We can derive the maximum likelihood estimate from the frequencies
of $\ifam{F}$.  Moreover, as the number of transactions goes to infinity, we
can derive the true distribution $p$ from the frequencies of $\ifam{F}$.
\end{theorem}

The preceding theorem showed that $\fm{M}$ is sufficient family of itemsets in
order to derive the correct true distribution. The next theorem shows that we
favor small families: if the data can be explained
with a simpler model, that is, using less itemsets, then the simpler model will
be chosen and, consequently, redundant itemsets will have a low score.

\begin{theorem}
\label{thr:occam}
Assume that data $D$ is generated from a distribution $p$ that comes from a
model $M$. Assume also that if any other model, say $M'$, contains this
distribution, then $\abs{\fm{M'}} > \abs{\fm{M}}$. Then the following holds: as
the number of data points in $D$ goes into infinity, $\score{X} = 1$ if $X \in
\fm{M}$, otherwise $\score{X} = 0$.
\end{theorem}

\subsection{Decomposable Models}
\label{sec:decompose}
We saw in Section~\ref{sec:compute} that in practice we cannot compute the score for
general exponential models. In this section we study a subclass of
exponential models, for which we can easily compute the needed score.  Roughly
speaking, a decomposable model is an exponential model where the corresponding
maximal itemsets can be arranged to a specific tree, called junction tree.  By
considering only decomposable models we obviously will lose some models, for example, the
discrete Gaussian model, that is, a model corresponding to all itemsets of size
1 and 2 is not decomposable. On the other hand, many interesting and
practically relevant models are decomposable, for example Chow-Liu trees.
Finally, these models are closely related to Bayesian networks and Markov
Random Fields (see~\cite{cowell:99:probabilistic} for more details).

To define a decomposable model,
let $\ifam{F}$ be a downward closed family of itemsets. We write $\ifam{G} =
\max\fpr{\ifam{F}}$ to be the set of maximal itemsets from $\ifam{F}$. Assume
that we can build a tree $\tree{T}$ using itemsets from $\ifam{G}$ as nodes with the
following property: If $X, Y \in \ifam{G}$ have a common item, say $a$, then
$X$ and $Y$ are connected in $\tree{T}$ (by a unique path) and every itemset
along that path contains $a$. If this property holds for $\tree{T}$, then
$\tree{T}$ is called \emph{junction tree} and $\ifam{F}$ is
\emph{decomposable}. We will use $E(\tree{T})$ to denote the edges of the tree.

Not all families have junction trees and some families may have multiple
junction trees.

\begin{example}
\begin{figure}[htb!]
\centering
\subfigure[Decomposable family $\ifam{F}_1$ of itemsets\label{fig:tree:a}]{\begin{tikzpicture}[>=latex',line join=bevel,scale=0.23]\pgfsetlinewidth{0.5bp}\input{pics/toy.tex}\end{tikzpicture}}
\hspace{0.3cm}
\begin{minipage}{4cm}
\subfigure[Decomposable family $\ifam{F}_2$ after merge\label{fig:tree:b}]{\begin{tikzpicture}[>=latex',line join=bevel,scale=0.23]\pgfsetlinewidth{0.5bp}\input{pics/toy2.tex}\end{tikzpicture}}
\subfigure[Decomposable family $\ifam{F}_3$ after the second merge\label{fig:tree:c}]{\begin{tikzpicture}[>=latex',line join=bevel,scale=0.23]\pgfsetlinewidth{0.5bp}\input{pics/toy3.tex}\end{tikzpicture}}
\end{minipage}
\caption{Figure~\ref{fig:tree:a} shows that the itemset family given in
Example~\ref{ex:toy} is decomposable. Figure~\ref{fig:tree:b} shows the junction tree for the family
after \textnormal{$\merge{\set{a_2}, a_3, a_4}$} and Figure~\ref{fig:tree:c} shows
the junction tree after \textnormal{$\merge{\set{a_4}, a_3, a_5}$}.}
\end{figure}
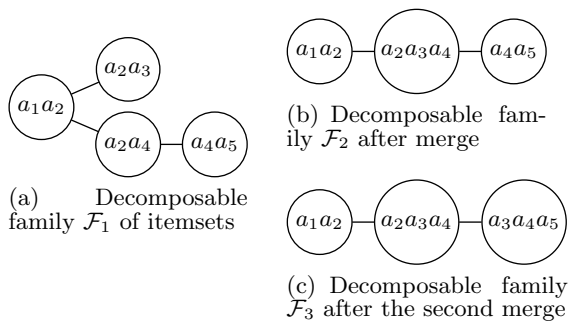

Let $\ifam{F}_1$ be the family of itemsets connected to the Chow-Liu model given
in Example~\ref{ex:toy}.
The maximal itemsets $\ifam{F}_1$ 
are $\set{a_1a_2, a_2a_3, a_2a_4, a_4a_5}$. Figure~\ref{fig:tree:a} shows
a junction tree for this family, making the family decomposable. On the
other hand, family $\set{a_1a_2, a_1a_3, a_2a_3}$ is not decomposable
since there is no junction tree for this family.
\end{example}

The most important property of decomposable families is that we can compute the
maximum likelihood efficiently. We first define the entropy of an itemset $X$, denoted by $\ent{X}$,
as
\[
	\ent{X} = -\sum_{t \in \set{0, 1}^{\abs{X}}} q_D(X = t) \log  q_D(X =t),
\]
where $q_D$ is the empirical distribution of the data.

\begin{theorem}
\label{thr:decomposelikelihood}
Let $\ifam{F}$ be a decomposable family and let $\tree{T}$ be its junction tree.
The maximum log-likelihood is equal to
\[
	- \frac{\log P(D \mid M, r^*)}{\abs{D}} = \sum_{X \in \max\fpr{\ifam{F}}} \ent{X} - \sum_{(X, Y) \atop \in E(\tree{T})} \ent{X \cap Y}.
\]
\end{theorem}

\begin{example}
Assume that our model space $\mathcal{M}$ consists only of two models, namely the tree model
$M_1$ given in Example~\ref{ex:toy} and the independence model, which we denote
$M_2$.
Assume also that we have a dataset with $9$ transactions,
\begin{alignat*}{3}
t_1 & = (1, 0, 0, 0, 0),\  &  t_2 & = (1, 1, 0, 0, 1),\  & t_3 & = (0, 0, 0, 0, 0), \\
t_4 & = (1, 1, 1, 1, 1),\  &  t_5 & = (1, 1, 0, 1, 1),\  & t_6 & = (0, 1, 1, 0, 0), \\
t_7 & = (0, 0, 1, 0, 0),\  &  t_8 & = (0, 0, 1, 0, 1),\  & t_9 & = (0, 1, 1, 1, 1). 
\end{alignat*}

To compute the probabilities $P(M_1 \mid D)$ and $P(M_2 \mid D)$, we need to know the entropies of certain itemsets
\[
\begin{split}
	\ent{a_1} = \ent{a_2} &= \ent{a_3} = \ent{a_5} = 0.68, \ent{a_4} = 0.64, \\
	\ent{a_1a_2} &= 1.31, \ent{a_2a_3} = 1.37, \\
	\ent{a_2a_4}& = \ent{a_4a_5} = 1.06.
\end{split}
\]

The log-likelihood of the independence model is equal to
\[
	\log P(D \mid M_2, r^*) = -9(4 \times 0.68 + 0.64).
\]
We use the junction tree given in Figure~\ref{fig:tree:a} and
Theorem~\ref{thr:decomposelikelihood} to compute the log-likelihood of $M_1$,
\[
	\log P(D \mid M_1, r^*) = -9(1.31 + 1.37 + 2 \times 1.06 - 2 \times 0.68 - 0.64).
\]
Note that $\abs{\fm{M_2}} = 6$ and $\abs{\fm{M_1}} = 10$.
Thus Eq.~\ref{eq:bic} implies that
\[
\begin{split}
	P(M_1 \mid D) &\propto P(D \mid M_1, r^*)\exp\fpr{-9/2\log 9} = 6.27 \times 10^{-14},\\
	P(M_2 \mid D) &\propto P(D \mid M_2, r^*)\exp\fpr{-5/2\log 9} = 2.43 \times 10^{-14}.
\end{split}
\]
We get the final probabilities by noticing that $P(M_1 \mid D) +
P(M_2 \mid D) = 1$ so that we have $P(M_1 \mid D) = 0.72$ and $P(M_2 \mid D) =
0.28$. Consequently, the scores for itemsets are equal to $\score{a_1a_2} = \score{a_2a_3} = \score{a_2a_4} = \score{a_4a_5} = 0.72$,
$\score{a_i} = 1$, for $i = 1, \ldots, 5$, and $\score{X} = 0$ otherwise.

\end{example}

%% file: pics/toy.tex
\pgfsetcolor{black}
  \draw [] (101.87bp,135.34bp) .. controls (116.47bp,141.71bp) and (132.46bp,148.69bp)  .. (147bp,155.04bp);
  \draw [] (248.28bp,53bp) .. controls (259.93bp,53bp) and (272.25bp,53bp)  .. (283.89bp,53bp);
  \draw [] (101.87bp,93.007bp) .. controls (116.15bp,86.873bp) and (131.76bp,80.168bp)  .. (146.04bp,74.031bp);
\begin{scope}
  \definecolor{strokecol}{rgb}{0.0,0.0,0.0};
  \pgfsetstrokecolor{strokecol}
  \draw (53bp,114bp) ellipse (53bp and 53bp);
  \draw (53bp,114bp) node {$a_1a_2$};
\end{scope}
\begin{scope}
  \definecolor{strokecol}{rgb}{0.0,0.0,0.0};
  \pgfsetstrokecolor{strokecol}
  \draw (195bp,53bp) ellipse (53bp and 53bp);
  \draw (195bp,53bp) node {$a_2a_4$};
\end{scope}
\begin{scope}
  \definecolor{strokecol}{rgb}{0.0,0.0,0.0};
  \pgfsetstrokecolor{strokecol}
  \draw (195bp,176bp) ellipse (52bp and 52bp);
  \draw (195bp,176bp) node {$a_2a_3$};
\end{scope}
\begin{scope}
  \definecolor{strokecol}{rgb}{0.0,0.0,0.0};
  \pgfsetstrokecolor{strokecol}
  \draw (337bp,53bp) ellipse (53bp and 53bp);
  \draw (337bp,53bp) node {$a_4a_5$};
\end{scope}

%% file: pics/toy2.tex
\pgfsetcolor{black}
  \draw [] (281.74bp,69bp) .. controls (293.76bp,69bp) and (306.07bp,69bp)  .. (317.57bp,69bp);
  \draw [] (106.39bp,69bp) .. controls (117.93bp,69bp) and (130.31bp,69bp)  .. (142.38bp,69bp);
\begin{scope}
  \definecolor{strokecol}{rgb}{0.0,0.0,0.0};
  \pgfsetstrokecolor{strokecol}
  \draw (53bp,69bp) ellipse (53bp and 53bp);
  \draw (53bp,69bp) node {$a_1a_2$};
\end{scope}
\begin{scope}
  \definecolor{strokecol}{rgb}{0.0,0.0,0.0};
  \pgfsetstrokecolor{strokecol}
  \draw (212bp,69bp) ellipse (69bp and 69bp);
  \draw (212bp,69bp) node {$a_2a_3a_4$};
\end{scope}
\begin{scope}
  \definecolor{strokecol}{rgb}{0.0,0.0,0.0};
  \pgfsetstrokecolor{strokecol}
  \draw (371bp,69bp) ellipse (53bp and 53bp);
  \draw (371bp,69bp) node {$a_4a_5$};
\end{scope}

%% file: pics/toy3.tex
\pgfsetcolor{black}
  \draw [] (281.56bp,69bp) .. controls (293.67bp,69bp) and (306.26bp,69bp)  .. (318.37bp,69bp);
  \draw [] (106.39bp,69bp) .. controls (117.93bp,69bp) and (130.31bp,69bp)  .. (142.38bp,69bp);
\begin{scope}
  \definecolor{strokecol}{rgb}{0.0,0.0,0.0};
  \pgfsetstrokecolor{strokecol}
  \draw (53bp,69bp) ellipse (53bp and 53bp);
  \draw (53bp,69bp) node {$a_1a_2$};
\end{scope}
\begin{scope}
  \definecolor{strokecol}{rgb}{0.0,0.0,0.0};
  \pgfsetstrokecolor{strokecol}
  \draw (212bp,69bp) ellipse (69bp and 69bp);
  \draw (212bp,69bp) node {$a_2a_3a_4$};
\end{scope}
\begin{scope}
  \definecolor{strokecol}{rgb}{0.0,0.0,0.0};
  \pgfsetstrokecolor{strokecol}
  \draw (388bp,69bp) ellipse (69bp and 69bp);
  \draw (388bp,69bp) node {$a_3a_4a_5$};
\end{scope}

%% file: algorithm.tex
\section{Sampling Models}
\label{sec:algorithm}

Now that we have means for computing the posterior probability of a single
decomposable model, our next step is to compute the score of an itemset
namely, the sum in Eq.~\ref{eq:defscore}. The problem is that this sum has
an exponential number of terms, and hence we cannot solve by enumerating
all possible families. We approach this problem from a different point of
view. Instead of computing the score for each itemset individually, we will
divide our mining method into two steps:
\begin{enumerate}
\item Sample random decomposable models from the posterior distribution $P(M \mid D)$.
\item Estimate the true score of an itemset by computing the number of sampled
families of itemsets in which the itemset occurs.
\end{enumerate}

\subsection{Moving from One Model to Another}

In order to sample we will use a MCMC approach by modifying the current
decomposable family by two possible operations, namely

\begin{itemize}
\item \textsc{Merge}: Select two maximal itemsets, say $X$ and $Y$.  Let $S =
X \cap Y$. Since $X$ and $Y$ are maximal, $X - S \neq \emptyset$ and $Y - S
\neq \emptyset$. Select $x \in X - S$ and $y \in Y - S$. Add a new itemset $S
\cup \set{x, y}$ into the family $\ifam{F}$ along with all possible
sub-itemsets.  We will use notation $\merge{S, x, y}$ to denote this operation.

\item\textsc{Split}: Select an itemset $X \in \max\fpr{\ifam{F}}$.
Select two items $x, y \in X$. Delete $X$ and all sub-itemsets containing $x$
and $y$ simultaneously. We will denote this operation by $\splt{X, x, y}$.
\end{itemize}

Naturally, not all splits and merges are legal, since some operations may
result in a family that is not decomposable, or even downward closed. 

\begin{example}
The family $\ifam{F}_2$ given in Figure~\ref{fig:tree:b} is obtained from the
family $\ifam{F}_1$ given in Figure~\ref{fig:tree:a} by performing
$\merge{\set{a_2}, a_3, a_4}$.  Moreover, $\ifam{F}_3$ (Figure~\ref{fig:tree:c}) is obtained from
$\ifam{F}_2$ by performing $\merge{\set{a_4}, a_3, a_5}$. Conversely, we can go
back by performing $\splt{a_3a_4a_5, a_3, a_5}$ first and $\splt{a_2a_3a_4,
a_3, a_4}$ second.
\end{example}

The next theorem tells us which splits are legal.

\begin{theorem}
\label{thr:legalsplit}
Let $\ifam{F}$ be decomposable family and let $X \in \max\fpr{\ifam{F}}$ and let
$x, y \in X$. Then the resulting family after a split operation $\splt{X, x, y}$ 
is decomposable if and only, there are no other maximal itemsets in
$\ifam{F}$ containing $x$ and $y$ simultaneously.
\end{theorem}

\begin{example}
All possible split combinations are legal in families $\ifam{F}_1$ and $\ifam{F}_2$ given in
Figure~\ref{fig:tree:a} and Figure~\ref{fig:tree:b}. However, for $\ifam{F}_3$
given in Figure~\ref{fig:tree:c}
$\splt{a_2a_3a_4, a_3, a_4}$ is illegal since
$a_3a_4a_5$ contains $a_3$ and $a_4$. Similarly, the operation $\splt{a_3a_4a_5, a_3, a_4}$
is illegal.
\end{example}

In order to identify legal merges, we will need some additional structures.  Let
$\ifam{F}$ be a downward closed family and let $\ifam{G} = \max\fpr{\ifam{F}}$
be its maximal itemsets. Let $S$ be an itemset.  We construct a \emph{reduced
family}, denoted by $\rf{\ifam{F}; S}$ with the following procedure.
Let us first define
\[
	\ifam{X} = \set{X - S \mid X \in \ifam{G}, S \subsetneq X}.
\]
To obtain the reduced family $\rf{\ifam{F}; S}$ from $\ifam{X}$, assume there
are two itemsets $X, Y \in \ifam{X}$ such that $X \cap Y \neq \emptyset$. We
remove these two sets from $\ifam{X}$ and replace them with $X \cup Y$. This is
continued until no such replacements are possible. We ignore any reduced family
that contains $0$ or $1$ itemsets. The reason for this will be seen in Theorem~\ref{thr:legalmerge},
which implies that such families will not induce any legal merges. 

\begin{example}
The non-trivial reduced families of the family given in
Figure~\ref{fig:tree:a} are $\rf{\ifam{F}_1; a_2} = \set{a_1, a_3, a_4}$ and
$\rf{\ifam{F}_1; a_4} = \set{a_2, a_5}$. Similarly, the reduced families for the
family given in Figure~\ref{fig:tree:b} are $\rf{\ifam{F}_2; a_2} = \set{a_1,
a_3a_4}$, and $\rf{\ifam{F}_2; a_4} = \set{a_2a_3, a_5}$. Finally, the reduced
families for the family given in Figure~\ref{fig:tree:c} are $\rf{\ifam{F}_3;
a_2} = \set{a_1, a_3a_4}$ and $\rf{\ifam{F}_3; a_3a_4} = \set{a_2, a_5}$.
\end{example}

The next theorem tells us when $\merge{S, x, y}$ is legal.

\begin{theorem}
\label{thr:legalmerge}
Let $\ifam{F}$ be decomposable family. A merge operation is legal, that is,
$\ifam{F}$ is still decomposable after adding $Z = S \cup \set{x, y}$ if and only
if there are sets $V, W \in \rf{\ifam{F}; S}$, $V \neq W$, such that $x \in V$
and $y \in W$.
\end{theorem}

\begin{example}
Family $\ifam{F}_2$ in Figure~\ref{fig:tree:b} is obtained from the family
$\ifam{F}_1$ in Figure~\ref{fig:tree:a} by $\merge{a_2, a_3a_4}$. This is legal
operation since $\rf{\ifam{F}_1; a_2} = \set{a_1, a_3, a_4}$. Similarly, merge
transforming $\ifam{F}_2$ to $\ifam{F}_3$ is legal since $\rf{\ifam{F}_2; a_4}
= \set{a_2a_3, a_5}$. However, this merge would not be legal in $\ifam{F}_1$
since we do not have $a_3$ in $\rf{\ifam{F}_1; a_4}$.
\end{example}

\subsection{MCMC Sampling Algorithm}

Sampling requires a proposal distribution $Q(M' \mid M)$.
Let $M$ be a current model. We denote the number of legal operations, either a
split or a merge, by $\dgr{M}$. Let $M'$ be a model obtained by sampling uniformly one
of the legal operations and applying it to $M$.
The probability of reaching $M'$ from $M$ with a single
step is $Q(M' \mid M) = 1 / \dgr{M}$. Similarly, the probability of reaching $M$ from $M'$
with a single step is $Q(M \mid M') = 1 / \dgr{M'}$.
Consequently, if we sample $u$ uniformly from the interval $[0, 1]$ and accept the step moving from $M$ into $M'$
if and only if $u$ is smaller than
\begin{equation}
	\frac{P(M' \mid D) Q(M \mid M')}{P(M \mid D)Q(M' \mid M)} = \frac{P(M' \mid D) \dgr{M}}{P(M \mid D)\dgr{M'}},
\label{eq:transstep}
\end{equation}
then the limit distribution of the MCMC will be the posterior distribution $P(M
\mid D)$ provided that the MCMC chain is ergodic. The next theorem shows that this is the case.

\begin{theorem}
\label{thr:ergodic}
Any decomposable model $M$ can be reached from any other model
$M'$ by a sequence of legal operations.
\end{theorem}

Our first step is to compute the ratio of the models given in Eq~\ref{eq:transstep}. To do that we
will use the BIC estimate given in Eq.~\ref{eq:bic} and
Theorem~\ref{thr:decomposelikelihood}. Let us first
define a function
\[
\begin{split}
	&\gain{X, x, y} = \\
	&\quad\abs{D}(\ent{X} -\ent{X - x} - \ent{X - y} + \ent{X - \set{x, y}}),
\end{split}
\]
where $X$ is an itemset and $x, y \in X$ are items.

\begin{theorem}
\label{thr:ratio}
Let $M$ be a decomposable model and let $M' = \splt{X, x, y ; M}$ be
a model obtained by a legal split. Let $A$ be the BIC estimate of $P(M \mid D)$
and let $B$ be the BIC estimate of $P(M' \mid D)$. Then
\[
	B/A = \exp\fpr{\gain{X, x, y} - \log \abs{D} 2^{\abs{X} - 3}}.
\]
Similarly, if $M' = \merge{S, x, y ; M}$, then
\[
	B/A = \exp\fpr{-\gain{S \cup \set{x, y}, x, y} + \log \abs{D} 2^{\abs{S} - 1}}.
\]
\end{theorem}

To compute the gain we need the entropies for 4 itemsets. Let $X$ be an
itemset. To compute $\ent{X}$ we first order the transactions in $D$ such that
the values corresponding to $X$ are in lexicographical order. This is done with
a radix sort $\textsc{Sort}(D, X)$ given in Algorithm~\ref{alg:sort}.  This
sort is done in $O(\abs{D}\abs{X})$ time.  After the data is sorted we can
easily compute the entropy with a single data scan: Set $e = 0$ and $p = 0$.
If the values of $X$ of the current transaction is equal to the previous
transaction we increase $p$ by $1/\abs{D}$, otherwise we add $-p \log p$ to $e$
and set $p$ to $1/\abs{D}$. Once the scan is finished, $e$ will be equal to $\ent{X}$.
The pseudo code for computing the entropy is given in Algorithm~\ref{alg:countentropy}.

\begin{algorithm}[htb!]
	\lIf{$X = \emptyset$ \OR $D = \emptyset$} {\Return $D$\;}
	$a_i \define $ first item in $X$\;
	$D_0 \define \set{t \in D \mid t_i = 0}$; $D_1 \define \set{t \in D \mid t_i = 1}$\;
	$D_0 \define \textsc{Sort}(D_0, X - a_i)$; $D_1 \define \textsc{Sort}(D_1, X - a_i)$\;
	\Return $D_0$ concatenated with $D_1$.
\caption{\textsc{Sort}$(D, X)$. Routine for sorting the transactions. Used 
by $\textsc{Entropy}$ as a pre-step for computing the entropy.} 
\label{alg:sort}
\end{algorithm}

\begin{algorithm}[htb!]
	$\textsc{Sort}(D, X)$\;
	$e \define 0$; $p \define 0$\;
	$u \define$ first transaction in $D$\;
	\ForEach{$t \in D$} {
		\If{$u_X \neq t_X$} {
			$e \define e - p \log p$\;
			$u \define t$\;
			$p \define 1 / \abs{D}$\;
		}
		\Else {
			$p \define p + 1 / \abs{D}$\;
		}
	}
	$e \define e - p \log p$\;
	\Return $e$\;
\caption{\textsc{Entropy}$(D, X)$. Computes the entropy of $X$ from the dataset $D$.} 
\label{alg:countentropy}
\end{algorithm}

Our final step is to compute $\dgr{M}$ and actually sample the operations. To do that we first write
$\sdgr{M}$ for the number of possible \textsc{Split} operations and let
$\sdgr{M, X}$ be the number of possible \textsc{Split} operations using itemset
$X$. Similarly, we write $\mdgr{M, S}$ for the number of legal merges using
$S$ and also $\mdgr{M}$ for the amount of legal merges in total.

Given a maximal itemset $X$ we build an occurrence table, which we denote by
$\tb{X}$, of size $\abs{X} \times \abs{X}$.  For $x, y \in X$, the entry of the table
$\tb{X, x, y}$ is the number of maximal itemsets containing $x$ and $y$.  If
$\tb{X, x, y} = 1$, then Theorem~\ref{thr:legalsplit} states that $\splt{X, x, y}$ is
legal. Consequently, to sample a split operation we first select a maximal
itemset weighted by $\sdgr{M, X} / \sdgr{M}$. Once $X$ is selected we select
uniformly one legal pair $(x, y)$. 

To sample legal merges, recall that $\merge{S,
x, y}$ involves selecting two maximal itemsets $X$ and $Y$ such that $S
= X \cap Y$, $x \in X - S$, and $y \in Y - S$.  Instead of selecting these
itemsets, we will directly sample an itemset $S$ and then select two items $x$
and $y$.  This sampling will work only if two legal merges 
$\merge{S_1, x_1, y_1}$ and $\merge{S_2, x_2, y_2}$ result in two different outcomes
whenever $S_1 \neq S_2$.

\begin{theorem}
\label{thr:reduceseparate}
Let $S_1$ and $S_2$ be two different itemsets and let $x_1, y_1 \notin S_1$,
and $x_2, y_2 \notin S_2$ be items. Assume that $\merge{S_i, x_i, y_i}$ is a
legal merge for $i = 1, 2$.  Define $Z_i = S_i \cup \set{x_i, y_i}$ for $i = 1,
2$. Then $Z_1 \neq Z_2$.
\end{theorem}

The construction of a reduced family states that, if $V, W \in
\rf{\ifam{F}; S}$, $V \neq W$, then $V \cap W = \emptyset$. It follows from
Theorem~\ref{thr:legalmerge} that
\[
	\mdgr{M, S} = \sum_{V, W \in \rf{\ifam{F}; S} \atop V \neq W} \abs{V}\abs{W}.
\]
To sample a merge we first sample an itemset $S$ weighted by $\mdgr{M,
S}/\mdgr{M}$.  Once $S$ is selected, we sample two different itemsets $V, W \in
\rf{\ifam{F}; S}$ (weighted by $\abs{V}$ and $\abs{W}$). Finally,
we sample $x \in V$ and $y \in W$.

Sampling $S$ for a merge operation is feasible only if the number of reduced
families for which the merge degree is larger than zero is small.

\begin{theorem}
Let $K$ be the number of items. There are at most $K$ maximal itemsets. There
are at most $K - 1$ itemsets for which the degree $\mdgr{M, \cdot} > 0$.
\label{thr:reducednumber}
\end{theorem}

Pseudo-code for a sampling step is given in Algorithm~\ref{alg:sample}.

\begin{algorithm}[htb!]
	$u \define$ random integer between 1 and $\dgr{M}$\;
	\If{$u \leq \sdgr{M}$} {
		Sample $X$ from $\max\fpr{\ifam{F}}$ weighted by $\sdgr{M, X}$\;
		Sample $x, y \in X$ such that $X$ is the only maximal itemset containing both $x$ and $y$\;
		$M' \define \splt{X, x, y; M}$\;
		\nllabel{alg:sample:gain1}$g \define \gain{X, x, y} - \log \abs{D} 2^{\abs{X} - 3}$\;
	}
	\Else {
		\nllabel{alg:sample:s}Sample $S$ weighted by $\mdgr{M, S}$\;
		Sample $V \in \rf{\ifam{F}; S}$ weighted by $\abs{V}$\;
		Sample $W \in \rf{\ifam{F}; S}$, $V \neq W$, weighted by $\abs{W}$\;
		Sample $x \in V$ and $y \in W$\;
		$M' \define \merge{S, x, y; M}$\;
		\nllabel{alg:sample:gain2}$g \define - \gain{S \cup \set{x, y}, x, y} + \log \abs{D} 2^{\abs{S} - 1}$\;
	}

	$z \define $ random real number from $[0, 1]$\;
	\lIf{$z \leq \exp\fpr{g}\frac{\dgr{M}}{\dgr{M'}}$} {\Return{$M'$}\; }
	\lElse {\Return{$M$}\; }
\caption{MCMC step for sampling decomposable models.}
\label{alg:sample}
\end{algorithm}

\subsection{Speeding Up the Sampling}

We have demonstrated what structures we need to compute
so that we can sample legal operations. After a sample, we can
reconstruct these structures from scratch. In 
this section we show how to optimize the sampling by constructing
the structures incrementally using
 Algorithms~\ref{alg:splitupdate}--\ref{alg:mergeside}.

First of all, we store only maximal itemsets of $\ifam{F}$. Theorem~\ref{thr:reducednumber} states
that there can be only $K$ such sets, hence split and merge operations can be
done efficiently.

During a split or a merge, we need to update what split operations are legal
after the split. We do this by updating an occurrence table $\tb{X}$. An update
takes $O(\abs{X}^2)$ time.  The next theorem shows which maximal itemsets we
need to update for legal split operations after a merge.

\begin{theorem}
\label{thr:splitupdate}
Let $\ifam{F}$ be a downward closed family of itemsets and let $\ifam{G}$ be
the family after performing $\merge{S, x, y}$. Let $Y$ be a maximal itemset in
$\max\fpr{\ifam{F}} \cap \max\fpr{\ifam{G}}$. Then legal split operations using $Y$ remain unchanged during
the merge unless $Y$ is the \emph{unique} itemset among maximal itemsets in $\ifam{F}$
containing either $S \cup \set{x}$ or $S \cup \set{y}$.
\end{theorem}

The following theorem tells us how reduced families should be updated after a
merge operation. To ease the notation, let us denote by $\link{\ifam{F}, S, x}$
the unique itemset (if such exists) in $\rf{\ifam{F}; S}$ containing $x$.

\begin{theorem}
\label{thr:mergeupdate}
Let $\ifam{F}$ be a downward closed family of itemsets and let $\ifam{G}$ be
the family after performing $\merge{S, x, y}$.
Then the reduced families are updated as follows:
\begin{enumerate}
\setlength{\parskip}{0.0em}
\item Itemsets $\link{\ifam{F}, S, x}$ and $\link{\ifam{F}, S, y}$ in $\rf{\ifam{F}; S}$
are merged into one itemset in $\rf{\ifam{G}; S}$.
\item Itemset $\set{x}$ is added into $\rf{\ifam{G}; S \cup \set{y}}$.
Itemset $\set{y}$ is added into $\rf{\ifam{G}; S \cup \set{x}}$.
\item Let $T \subsetneq S$ and let $z \in S - T$. 
The itemset containing $z$ in $\rf{\ifam{F}; T \cup \set{y}}$ is augmented with item $x$.
Similarly, itemset containing $z$ in $\rf{\ifam{F}; T \cup \set{x}}$ is augmented with item $y$.
\item Otherwise, $\rf{\ifam{F}; T} = \rf{\ifam{G}; T}$ or $\mdgr{\ifam{F}; T} = 0$ and $\mdgr{\ifam{G}; T} = 0$.
\end{enumerate}
\end{theorem}

Theorems~\ref{thr:splitupdate}~and~\ref{thr:mergeupdate} only covered the
updates during merges.  Since $\splt{S \cup \set{x, y}, x, y}$ and $\merge{S,
x, y}$ are opposite operations we can derive the needed updates for splits from
the preceding theorems.

\begin{corollary}[of Theorem~\ref{thr:splitupdate}]
\label{cor:splitupdate}
Let $\ifam{F}$ be a downward closed family of itemsets and let $\ifam{G}$ be
the family after performing $\splt{X, x, y}$. Let $Y$ be a maximal itemset in
$\max\fpr{\ifam{F}} \cap \max\fpr{\ifam{G}}$. Then legal split operations using $Y$ remain unchanged during
the merge unless $Y$ is the \emph{unique} itemset among maximal itemsets in $\ifam{G}$
containing either $X - \set{x}$ or $X - \set{y}$.
\end{corollary}

\begin{corollary}[of Theorem~\ref{thr:mergeupdate}]
\label{cor:mergeupdate}
Let $\ifam{F}$ be a downward closed family of itemsets and let $\ifam{G}$ be
the family after performing $\splt{X, x, y}$. Let $S = X - \set{x, y}$.
Then the reduced families are updated as follows:
\begin{enumerate}
\setlength{\parskip}{0.0em}
\item Itemset containing $\set{x, y}$ in $\rf{\ifam{F}; S}$ is split into two parts,
$\link{\ifam{G}, S, x}$ and $\link{\ifam{G}, S, y}$.
\item Itemset $\set{x}$ is removed from $\rf{\ifam{G}; S \cup \set{y}}$.
Itemset $\set{y}$ is removed from $\rf{\ifam{G}; S \cup \set{x}}$.
\item Let $T \subsetneq S$ and let $z \in S - T$. Item $x$ is removed from
the itemset containing $z$ in $\rf{\ifam{F}; T \cup \set{y}}$.
Similarly, item $y$ is removed from the itemset containing $z$ in $\rf{\ifam{F}; T \cup \set{x}}$.
\item Otherwise, $\rf{\ifam{F}; T} = \rf{\ifam{G}; T}$ or $\mdgr{\ifam{F}; T} = 0$ and $\mdgr{\ifam{G}; T} = 0$.
\end{enumerate}
\end{corollary}

We keep in memory only those families that have
positive merge degree. Theorem~\ref{thr:reducednumber} tells us that there are
only $K - 1$ such families. By studying the code in the update algorithm we see
that, except in two cases, the update of a family is either a
insertion/deletion of an element into an itemset or a merge of two itemsets.
The first complex case is given on Line~\ref{alg:mergeside:complex1} in
\textsc{MergeSide} which corresponds to Case 2 in
Theorem~\ref{thr:mergeupdate}. The problem is that this family may have
contained only itemset before the merge, hence we did not store it. Consequently, we
need to recreate the missing itemset, and this is done in $O(\sum_{X \in
\ifam{F}} X)$ time.  The second case occurs on
Line~\ref{alg:mergeside:complex2} in \textsc{SplitSide}.  This corresponds to
the case where we need to break the itemset $W \in \rf{S}$ containing $x$ and
$y$ apart during a split (Case 1 in Corollary~\ref{cor:mergeupdate}). This is done by constructing the new sets from
scratch.  The construction needs $O(\abs{W}K\pr{\abs{W} + M})$ time, where $M$
is the size of largest itemset in $\ifam{F}$.

\begin{algorithm}[htb!]
	Update $\ifam{F}$\;
	Remove $\link{S, x}$ from $\rf{S}$\;
	$S \define X - \set{x, y}$\;
	$\textsc{SplitSide}(X, x, y, S)$; $\textsc{SplitSide}(X, y, x, S)$\;
\caption{$\textsc{SplitUpdate}(X, x, y)$. Routine for updating the structures during $\splt{X, x, y}$.}
\label{alg:splitupdate}
\end{algorithm}

\begin{algorithm}[htb!]
	$Z \define S \cup \set{a}$\;
	\While{changes}  {
		\nllabel{alg:mergeside:complex2}
		$Z \define Z \cup \set{X \in \max\fpr{\ifam{F}}; S \subsetneq Z \cap X}$\;
	}
	Add $Z$ into $\rf{S}$\;
	\If{$\rf{S \cup \set{a}}$ exists} {
		Remove $\set{b}$ from $\rf{S \cup \set{a}}$\;
	}
	\For{$T \subsetneq S$, $\rf{T \cup \set{a}}$ exists} {
		$z \define$ (any) item in $S - T$\;
		Remove $b$ from $\link{S, z}$\;
	}
	\If{there is unique $Z \in \max\fpr{\ifam{F}}$ s.t. $S \cup \set{a} \subsetneq Z$} {
		Update $\tb{Z}$\;
	}
\caption{Subroutine $\textsc{SplitSide}(X, a, b, S)$ used by \textsc{SplitUpdate}.}
\label{alg:splitside}
\end{algorithm}

\begin{algorithm}[htb!]
	Merge $\link{S, x}$ and $\link{S, y}$ in $\rf{S}$\;
	$A \define$ itemset in $\max\fpr{\ifam{F}}$ such that $S \cup \set{x} \subseteq A$\;
	$B \define$ itemset in $\max\fpr{\ifam{F}}$ such that $S \cup \set{y} \subseteq B$\;
	Build $\tb{S \cup \set{x, y}}$ from $\tb{A}$ and $\tb{B}$\; 
	$\textsc{MergeSide}(S, x, y)$; $\textsc{MergeSide}(S, y, x)$\;
	Update $\ifam{F}$\;
\caption{$\textsc{MergeUpdate}(S, x, y)$. Routine for updating the structures during $\merge{S, x, y}$.}
\label{alg:mergeupdate}
\end{algorithm}

\begin{algorithm}[htb!]
	$U \define S \cup \set{a}$\;
	\If{$U \notin \max\fpr{\ifam{F}}$ \AND $\rf{U}$ does not exists} {
		\nllabel{alg:mergeside:complex1}$\rf{U} \define \set{\bigcup_{U \subseteq X \in \max\fpr{\ifam{F}}}X}$\;
	}
	Add $b$ into $\rf{U}$\;
	\For{$T \subsetneq S$, $\rf{T \cup \set{a}}$ exists} {
		$z \define$ (any) item in $S - T$\;
		Augment $\link{S, z}$ with $b$\;
	}
	\If{there is unique $Z \in \max\fpr{\ifam{F}}$ s.t. $U \subsetneq Z$} {
		Update $\tb{Z}$\;
	}
\caption{Subroutine $\textsc{MergeSide}(S, a, b)$ used by \textsc{MergeUpdate}.}
\label{alg:mergeside}
\end{algorithm}

%% file: related.tex
\section{Related Work}
\label{sec:related}
Many quality measures have been suggested for itemsets. A major part of these
measures are based on how much the itemset deviates from some null hypothesis. For
example, itemset measures that use the independence model as background
knowledge have been suggested in~\cite{aggarwal:98:new,brin:97:beyond}. More
flexible models have been proposed, such as, comparing itemsets against graphical
models~\cite{jaroszewicz:04:interestingness} and local Maximum Entropy
models~\cite{meo:00:theory,tatti:08:maximum}. In addition, mining itemsets with
low entropy has been suggested in~\cite{heikinheimo:07:finding}.

Our main theoretical advantage over these approaches is that we
look at the itemsets as a whole collection. For example, consider that we
discover that item $a$ and $b$ deviate greatly from the null hypothesis. Then
any itemset containing both $a$ and $b$ will also be deemed interesting. The
reason for this is that these methods are not adopting to the discovered fact that $a$ and $b$ are
correlated, but instead they continue to use the same null hypothesis.  We, on
the other hand, avoid this problem by considering models: if
itemset $ab$ is found interesting that information is added into the
statistical model. If this new model then explains bigger itemsets containing
$a$ and $b$, then we have no reason to add these itemsets, into the model, and hence such
itemsets will not be considered interesting.

The idea of mining a pattern set as a whole in order to reduce the number of
patterns is not new. For example, pattern reduction techniques based on minimum
description length principle has been
suggested~\cite{siebes:06:item,tatti:08:finding,heikinheimo:09:low-entropy}.
Discovering decomposable models have been studied in~\cite{tatti:08:decomposable}.
In addition, a framework that incrementally adopts to the patterns approved
by the user has been suggested in~\cite{hanhijarvi:09:tell}. Our main advantage
is that these methods require already discovered itemset collection as an input,
which can be substantially large for low thresholds. We, on the other hand,  skip this step and define
the significance for itemsets such that we can mine the patterns directly.

%% file: experiments.tex
\section{Experiments}
\label{sec:experiments}
In this section we present our empirical evaluation of the measure.  We first
describe the datasets and the setup for the experiments, then present the
results with synthetic datasets and finally the results with real-world
datasets.

\subsection{Setup for the Experiments}
We used $2 \times 3$ synthetic datasets and $3$ real-world datasets.

The first three synthetic datasets, called \emph{Ind}, contained $15$
independent items and $100$, $10^3$, and $10^4$ transactions, respectively. We
set the frequency for the individual items to be $0.1$. The next three synthetic
datasets, called \emph{Path}, also contained $15$ items. In these datasets, an
item $a_i$ were generated from the previous one with $P(a_i = 1 \mid a_{i - 1}
= 1) = P(a_i = 0 \mid a_{i - 1} = 0) = 0.75$.  The probability of the first
item was set to $0.5$. We set the number of transactions for these datasets to
$100$, $10^3$, and $10^4$, respectively. 

Our first real-world dataset \emph{Paleo}\footnote{NOW public release 030717
available from~\cite{fortelius:06:spectral}.}  contains information of species
fossils found in specific paleontological sites in
Europe~\cite{fortelius:06:spectral}.  The dataset \emph{Courses} contains the
enrollment records of students taking courses at the Department of Computer
Science of the University of Helsinki. Finally, our last dataset is \emph{Dna}
is DNA copy number amplification data collection of human
neoplasms~\cite{myllykangas:06:dna}. We used $100$ first items from this data
and removed empty transactions. The basic characteristics of the datasets are
given in Table~\ref{tab:basic}.

For each data we sampled the models from the posterior distribution using
techniques described in Section~\ref{sec:decompose}. We used singleton model
as a starting point and did $5000$ restarts. The number of required MCMC steps
is hard to predictm since the structure of the state space of decomposable
models is complex. Further it also depends on the actual data. Hence, we settle for heuristic:
for each restart we perform $100
K \log K$ MCMC steps, where $K$ is the number of items. 
Doing so we obtained $N = 5000$ random models for each dataset.  The
execution times for sampling are given in Table~\ref{tab:basic}.  Let
$\enset{\ifam{F}_1}{\ifam{F}_N}$ be the discovered models.  We estimated the
itemset score $\score{X} \approx \abs{\set{\ifam{F}_i \mid X \in \ifam{F}_i}} /
N$ and mined interesting itemsets using a simple depth-first approach.

\begin{table}[htb!]
\centering
\begin{tabular}{lrrrr}
\toprule
Name & $\abs{D}$ & K & \# of steps & time \\
\midrule
\emph{Ind}     & $100$--$10^4$ & 15  & $4\,063$  & $4m$--$2h$\\
\emph{Path}    & $100$--$10^4$ & 15  & $4\,063$  & $7m$--$3.5h$\\
\emph{Dna}     & 1160          & 100 & $46\,052$ & $6h$\\
\emph{Paleo}   & 501           & 139 & $68\,590$ & $5h$\\
\emph{Courses} & 3506          & 90  & $40\,499$ & $8.5h$\\
\bottomrule
\end{tabular}
\caption{Basic characteristics of the datasets.
The fourth column contains
the number of sample steps and the last column is the execution time.}
\label{tab:basic}
\end{table}

\subsection{Synthetic datasets}
Our main purpose for the experiments with synthetic data\-sets is to demonstrate
how the score behaves as a function of number of data points.  To this end, we
plotted the number of significant itemsets, that is itemsets whose score was
higher than the threshold $\sigma$, as a function of the threshold $\sigma$.
The results are shown in
Figures~\ref{fig:ind_sig_vs_thresh}~and~\ref{fig:path_sig_vs_thresh}.
\begin{figure*}[htb!]
\centering
\subfigure[\label{fig:ind_sig_vs_thresh}\emph{Ind}]{\includegraphics{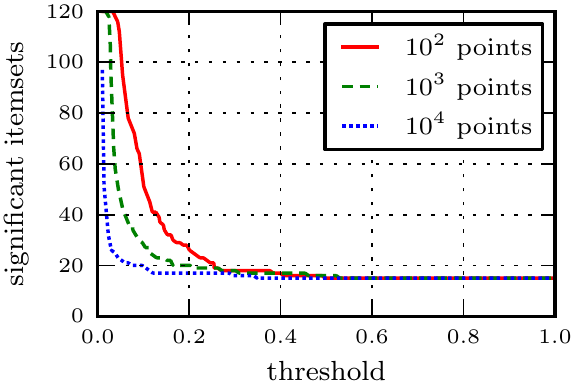}}%
\subfigure[\label{fig:path_sig_vs_thresh}\emph{Path}]{\includegraphics{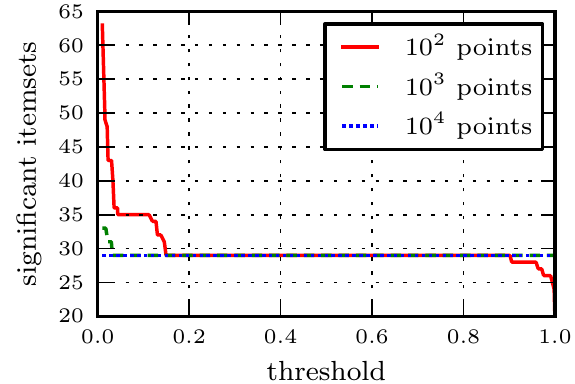}}%
\subfigure[\label{fig:real_sig_vs_thresh}Real datasets]{\includegraphics{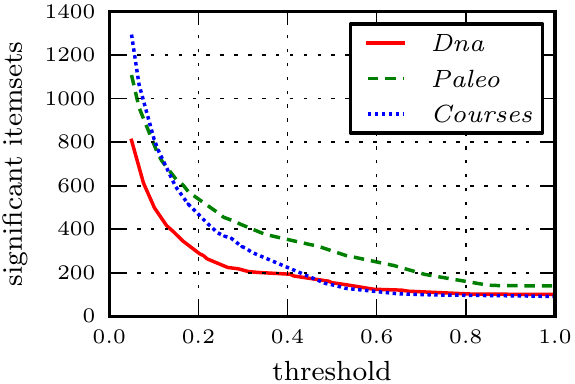}}
\caption{Number of significant itemsets as a function of the threshold.
The smallest threshold used for \emph{Ind} and \emph{Path} is $0.01$.
The smallest threshold used for real datasets is $0.05$.}
\end{figure*}

Ideally, for \emph{Ind}, the dataset with independent variables we should have
only $15$ significant itemsets, that is, the singletons, for any $\sigma > 0$. Similarly, for
\emph{Path} we should have $15 + 14 = 29$ itemsets, the singletons and the
pairs of form $a_ia_{i + 1}$. We can see from
Figures~\ref{fig:ind_sig_vs_thresh}~and~\ref{fig:path_sig_vs_thresh} that as we
increase the number of transactions in data, the number of significant itemsets
approaches these ideal cases, as predicted by Theorem~\ref{thr:occam}.  The
convergence to the ideal case is faster in \emph{Path} than in \emph{Ind}. The
reason for this can be explained by the curse of dimensionality. In \emph{Ind}
we have $15 \times 14 / 2 = 105$ combinations of pairs of items. There is a
high probability that some of these item pairs appear to be correlated. On the
other hand, for \emph{Path}, let us assume that we have the correct model. That
is, the singletons and the $14$ pairs $a_ia_{i + 1}$. The only valid itemsets
of size $3$ that we can add to this model are of the form $a_ia_{i + 1}a_{i + 2}$.
There are only $13$ of such sets, hence the probability of finding such itemset
important is much lower. Interestingly, in \emph{Path} we actually benefit from
the fact that we are using decomposable models instead of general exponential
models.

\subsection{Use cases with real-world datasets}

Our first experiment with real-world data is to study the
number of significant itemsets as a function of the threshold $\sigma$.
Figure~\ref{fig:real_sig_vs_thresh} shows the number of significant itemsets
for all three datasets. We see that the number of significant
itemsets increases faster than for the synthetic datasets as the threshold
decreases. The main reason for this difference, is that with real-world datasets we
have more items and less transactions. This is seen especially in the \emph{Paleo}
dataset for which the number of significant itemsets increases steeply between the
interval $0.4$--$1.0$ when compared to \emph{Dna} and \emph{Courses}.

Our next experiment is to compare the score against baselines, namely, the
frequency $\freq{X}$ and entropy $\ent{X}$. These comparisons are given in
Figures~\ref{fig:real_score_vs_freq}~and~\ref{fig:real_score_vs_entropy}.  In
addition, we computed the correlation coefficients (given in
Table~\ref{tab:corr}). From results, we see that $\score{X}$ has a positive
correlation with frequency and a negative correlation with entropy. The
correlation with entropy is expected, since low-entropy implies that the empirical
distribution of an itemset is different than the uniform distribution. Hence,
using the frequency of such an itemset should improve the model and
consequently the itemset is considered interesting.

\begin{figure*}[htb!]
\centering
\subfigure[score as a function of frequency\label{fig:real_score_vs_freq}]{\includegraphics[width=2.3in]{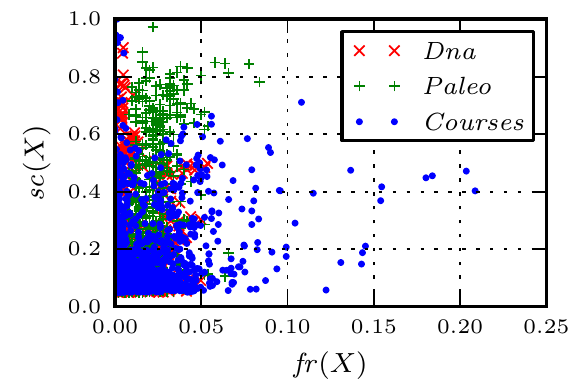}}%
\subfigure[score as a function of entropy\label{fig:real_score_vs_entropy}]{\includegraphics[width=2.3in]{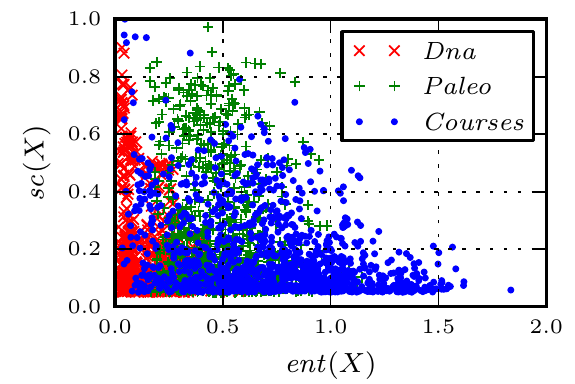}}%
\subfigure[index difference vs. score\label{fig:real_dist_vs_score}]{\includegraphics[width=2.3in]{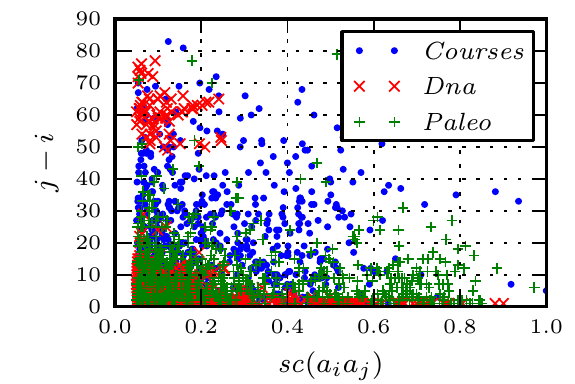}}

\caption{Score $\score{X}$ as a function of different baselines. In
Figures~\ref{fig:real_score_vs_freq}~and~\ref{fig:real_score_vs_entropy} score
is plotted as a function of frequency $\freq{X}$ and entropy $\ent{X}$,
respectively.  In Figure~\ref{fig:real_dist_vs_score} the difference between
the largest index and the smallest index of an item is plotted as a function of
$\score{X}$. Singletons and itemsets with $\score{X} < 0.05$ are ignored. Additionally, only itemsets of size $2$ are used
in Figure~\ref{fig:real_dist_vs_score}.}
\label{fig:real_score_vs_baseline}
\end{figure*}

\begin{table}[htb!]
\centering
\begin{tabular}{lrrr}
\toprule
Data & $\freq{X}$ & $\ent{X}$ & index diff. \\
\midrule
\emph{Dna}     & 0.16 & -0.27 & -0.28\\
\emph{Paleo}   & 0.45 & -0.16 & -0.17\\
\emph{Courses} & 0.18 & -0.35 & -0.02\\
\bottomrule
\end{tabular}
\caption{Correlations of $\score{X}$ and baselines.}
\label{tab:corr}
\end{table}

Both datasets \emph{Paleo} and \emph{Dna} have a natural order between the
items, for example, for \emph{Paleo} dataset it is the era when particular
species was extant.  Our next experiment is to test whether this order is represented
in the discovered patterns. In our experiments, the items in these datasets
were ordered using this natural order.
Let $X = a_ia_j$ be an itemset of size $2$.
In Figure~\ref{fig:real_dist_vs_score}, we plotted $j - i$ as a function of
$\score{X}$. We used only itemsets of size $2$, since larger itemsets have
inherently smaller scores and larger difference between the indices. In
addition, we compute the correlation coefficients (given in Table~\ref{tab:corr}).
From the results we see that for \emph{Paleo} and \emph{Dna} the more significant
itemsets tend to have a smaller difference between their indices. On the other
hand, we do not discover any significant correlation in \emph{Courses}.

Finally, we report some of the discovered patterns from \emph{Courses}.
The $4$ most significant itemsets of size $2$ are (\emph{Computer
Architectures}, \emph{Performance Analysis}) with a score of $0.95$, (\emph{Design \&
Analysis of Algorithms}, \emph{Principles of Functional Programming}) scoring
$0.94$, (\emph{Database Systems II}, \emph{Information Storage}) scoring
$0.94$, (\emph{Three concepts: probability}, \emph{Machine Learning}) scoring
$0.92$. 

%% file: conclusions.tex
\section{Conclusions}
\label{sec:conclusions}
In this paper we introduced a novel and general approach for ranking itemsets.
The idea behind the approach is to connect statistical models and collections
of itemsets. This connection enables us to define the score of an itemset as the
probability of itemset occurring in a random model. Doing so, we transformed
the problem of mining patterns into a more classical problem of modeling.

As a concrete example of the framework, we used exponential models.
These models have many important theoretical and practical properties. The connection with itemsets is
natural and the Occam's razor inherent to these models can be used 
against the pattern explosion problem. Our experiments support the theoretical results
and demonstrate that the measure works in practice.

%% file: appendix.tex
\appendix
\section{Proofs for Section~\lowercase{\ref{sec:model}}}

\begin{proof}[of Theorem~\ref{thr:itemsetsderive}]
We will prove the theorem using the well-known connection between maximum
entropy principle and exponential models.
Let $q_D$ be the empirical distribution of the data.
Let $\mathcal{P}$ be the collection
of distributions,
\[
	\mathcal{P} = \set{p \mid p(X = 1) = \freq{X} = q_D(X = 1), X \in \ifam{F}},
\]
that is, a distribution in $\mathcal{P}$ has the same itemset frequencies for
$\ifam{F}$ as the data. 
Assume for the time being that all entries in $q_D$ are positive. Then, A
well-known theorem~\cite{csiszar:75:i-divergence} states that the unique
distribution, say $p^*$, maximizing the entropy among $\mathcal{P}$ belongs to
the exponential model $M$ such that $\fm{M} = \ifam{F}$. Let the corresponding
parameters for the model be $r^*$.

We need to show that $p^*$ is actually a maximum likelihood distribution. To see this, let
$q \in M$ be another distribution from the model and let $r$ be its parameters.
A straightforward calculus reveals that
\[
\begin{split}
	\log \frac{p^*\fpr{D}}{q\fpr{D}} & = \sum_{t \in D} \log \frac{p^*(t)}{q(t)} \\
	                                 & = \abs{D}\sum_{t \in \set{0, 1}^K} q_D(t)\log \frac{p^*(t)}{q(t)} \\
	                                 & = \abs{D}\sum_{t \in \set{0, 1}^K} \sum_{X \in \ifam{F}} q_D(X = 1) (r^*_X - r_X) \\
	                                 & = \abs{D}\sum_{t \in \set{0, 1}^K} \sum_{X \in \ifam{F}} p^*(X = 1) (r^*_X - r_X) \\
	                                 & = \abs{D}\sum_{t \in \set{0, 1}^K} p^*(t)\log \frac{p^*(t)}{q(t)} \\
	                                 & = \abs{D}\kl{p^* \| q} \geq 0,
\end{split}
\]
where $\kl{p \| q}$ is the Kullback-Leibler divergence between $p^*$ and $q$.
The inequality shows that $r^*$ has the highest likelihood.  If $q_D$ has zero
probabilities, then we can consider a sequence $\epsilon_n \to 0$ as $n \to
\infty$. Define $q_n = (1 - \epsilon_n)q_D + \epsilon_n 2^{-K}$. Let $p^*_n$ be
the maximal entropy distribution computed from $q_n$. It is easy to show that a
limit distribution of any converging subsequence of $p^*_n$ has the
maximal entropy among $\mathcal{P}$. Since maximal entropy distribution is unique,
the sequence $p^*_n$ converges to the unique value, say $p^*$. Following the same
calculus as above we can show that $p^*$ is the maximal likelihood distribution
\[
\begin{split}
&\sum_{t \in \set{0, 1}^K} q_D(t)\log \frac{p^*(t)}{q(t)}   \\
	                                 &\qquad = \lim_{n \to \infty} \sum_{t \in \set{0, 1}^K} q_n(t)\log \frac{p^*_n(t)}{q(t)}  \\
	                                 &\qquad = \lim_{n \to \infty} \sum_{t \in \set{0, 1}^K} \sum_{X \in \ifam{F}} q_n(X = 1) (r^*_X - r_X) \\
	                                 &\qquad = \lim_{n \to \infty} \sum_{t \in \set{0, 1}^K} \sum_{X \in \ifam{F}} p^*_n(X = 1) (r^*_X - r_X) \\
	                                 &\qquad = \lim_{n \to \infty} \sum_{t \in \set{0, 1}^K} p^*_n(t)\log \frac{p^*_n(t)}{q(t)} \\
	                                 &\qquad = \sum_{t \in \set{0, 1}^K} p^*(t)\log \frac{p^*(t)}{q(t)} = \kl{p^* \| q} \geq 0.
\end{split}
\]
Finally, as the number of transactions goes into infinity $r^*$ converges to
the true parameters of the model and $p^*$ converges into true underlying
distribution $p$.
\end{proof}

To prove Theorem~\ref{thr:occam} we will need the following lemma whose
prove can be found for example in~\cite{wald:43:tests}.

\begin{lemma}
\label{lem:x2}
Assume that $D$ with $N$ data points comes from a model $M$ with parameters
$r$. Let $r^N$ be the maximal likelihood estimate. Then
\[
	\log P(D \mid M, r^N) - \log P(D \mid M, r)
\]
converges uniformly to $\chi^2$ distribution with $d \leq \abs{F} - 1$ number
degrees of freedom.
\end{lemma}

\begin{proof}[of Theorem~\ref{thr:occam}]

Let $M' \neq M$ be a model and $\ifam{F}' = \fm{M'}$. We need to show that
\[
	\frac{P(M' \mid D)}{P(M \mid D)} \to 0
\]
as the number of data points goes to
infinity.

We will use the BIC estimate given in Eq.~\ref{eq:bic} to prove the theorem.
Let $N = \abs{D}$. Let us define $p^N$ to be the maximal likelihood distribution
from $M$ and $q^N$ to be the maximal likelihood distribution in $M'$.
Let us write 
\[
	A_N = \log q^N(D) - \log p^N(D)
\]
and
\[
	B_N = \log N (\abs{\ifam{F}} - \abs{\ifam{F}'})/2.
\]
Then asymptotically
\[
	\log P(M' \mid D) - \log P(M \mid D) = A_N + B_N.
\]
Let $A$ be the limit of $A_N$ and $B$ be the limit of $B_N$.

Assume that $p \notin M'$. Let $q \in M'$ be the maximum likelihood
distribution as $N \to \infty$. Since the exponential models are closed sets
$q$ cannot be $p$.  A straightforward calculation reveals that $N^{-1}A_N \to -
\kl{p \| q} < 0$. In addition, $N^{-1}B_N \to 0$.  Thus $N^{-1}(A_N + B_N)$
approaches a negative value, and so $A + B = -\infty$.

Assume that $p \in M'$, then the conditions imply that $\abs{\ifam{F}'} >
\abs{\ifam{F}}$ so that $B = -\infty$.  Lemma~\ref{lem:x2} implies that $A$ is
a difference of two $\chi^2$ distributions. More importantly, since convergence
in Lemma~\ref{lem:x2} is uniform we can compute the limit inside the
probability so that the probability 
\[
	P(A_N + B_N > \sigma) \to P(A + B > \sigma) = P(A > \infty) = 0
\]
for any $\sigma \in \real$.  Hence $A_N + B_N$
approaches $-\infty$. This proves the theorem.
\end{proof}

\begin{proof}[of Theorem~\ref{thr:decomposelikelihood}]
Let $p(A) = P(A \mid M, r)$ be the distribution from a model $M$ with parameters $r$.
To ease the notation, let us define $\ifam{G} = \max\fpr{\ifam{F}}$ and 
\[
	\ifam{S} = \set{X \cap Y \mid (X, Y) \in E(\tree{T})}.
\]
A classic result states that we can decompose $p$ to factors,
\[
	p(A = t) = \frac{\prod_{X \in \ifam{G}} p(X = t_X)}{\prod_{S \in \ifam{S}} p(S = t_S)},
\]
where $t_X$ is a projection of a binary vector $t$ to the variables of $X$.
Let $q(A)$ be the empirical distribution. Let $p$ be the maximum-likelihood
distribution.  The proof of Theorem~\ref{thr:itemsetsderive} states that $p(X =
1) = q(X = 1)$ for any itemset $X \in \ifam{F}$. Using the exclusion-inclusion
rules and the fact that $\ifam{F}$ is downward closed we can show that $p(X =
t) = q_D(X = t)$ for every $X \in \ifam{F}$ and every possible binary vector
$t$ of length $\abs{X}$.  The log-likelihood is equal to
\[
\begin{split}
	\log \prod_{t \in D} p(A = t) = & \sum_{t \in D} \log \frac{\prod_{X \in \ifam{G}} p(X = t_X)}{\prod_{S \in \ifam{S}} p(S = t_S)} \\
	                              = & \sum_{X \in \ifam{G}} \sum_{t \in D} \log p(X = t_X)  \\
								    & \  - \sum_{S \in \ifam{S}} \sum_{t \in D} \log p(S = t_{S})) \\
	                              = & \abs{D}\sum_{X \in \ifam{G}} \sum_{t_X} q(X = t_X) \log p(X = t_X)  \\
								    & \  - \abs{D}\sum_{S \in \ifam{S}} \sum_{t_S} q(S = t_S) \log p(S = t_{S})) \\
	                              = & \abs{D}\sum_{X \in \ifam{G}} \sum_{t_X} q(X = t_X) \log q(X = t_X)  \\
								    & \  - \abs{D}\sum_{S \in \ifam{S}} \sum_{t_S} q(S = t_S) \log q(S = t_{S})) \\
	                              = & -\abs{D}\sum_{X \in \ifam{G}} \ent{X}  + \abs{D}\sum_{S \in \ifam{S}} \ent{S}. \\
\end{split}
\]
This proves the theorem.
\end{proof}

\section{Proofs for Section~\lowercase{\ref{sec:algorithm}}}

We will need the following technical lemma for several subsequent proofs.

\begin{lemma}
\label{lem:cycle}
Let $X_1, \ldots, X_N \in \max\fpr{\ifam{F}}$ with $N \geq 3$. Define $S_i =
X_i \cap X_{i + 1}$ and $S_N = X_N \cap X_1$. If each $S_i$ has a unique item
(when compared to other $S_i$), then $\ifam{F}$ is not decomposable.
\end{lemma}

\begin{proof}[of Theorem~\ref{thr:legalsplit}]
Let us assume that $X$ is the only itemset containing $x$ and $y$ simultaneously.  Let
$\tree{T}$ be a junction tree. The adjacent itemsets of $X$ either miss $x$ or
$y$ or both. Thus we can remove $X$ and replace it $X - x$ and $X - y$,
connected to each other, the adjacent itemsets of $X$ can be connected to
either $X - x$ or to $X - y$. If $X - x$ is not maximal, that is there is an
itemset $Z \supset X - x$ , then we can remove $X - x$ and connect all the
edges going to $X - x$ to $Z$. We can repeat this for $X - y$ as well.  The
resulting tree is a junction tree.

To prove the other direction, assume that there is an another itemset $Y$
containing $x$ and $y$, simultaneously. There must be an item $z \in X$ such
that $z \notin Y$.  Then itemsets $Y$, $X - x$, and $X - y$ satisfy the
requirements of Lemma~\ref{lem:cycle}, hence the new family is not
decomposable.
\end{proof}

\begin{proof}[of Theorem~\ref{thr:legalmerge}]
Let $\tree{T}$ be a junction tree.  Assume that $V$ and $W$ exist. By
construction of $\rf{\ifam{F}; S}$ it follows that there must be two itemsets
$X, Y \in \ifam{F}$ such that $S = X \cap Y$, $x \in X$, and $y \in Y$.  If $X$
and $Y$ are adjacent in $\tree{T}$, we can add $Z$ between $X$ and $Y$
(possibly removing $X$ and $Y$ if they are no longer maximal), and $\tree{T}$
still remains a junction tree.

If $X$ and $Y$ are not adjacent, then there is a path $P$ connecting them.
There must be itemsets $P_i$ and $P_{i + 1}$ along that path such that $P_i
\cap P_{i + 1} = S$, otherwise $X$ and $Y$ would have been merged during the
construction of reduced family, which is a contradiction. If we remove edge
$(P_i, P_{i + 1})$ and add edge $(X, Y)$ we will obtain an alternative junction
tree for $\ifam{F}$. It is straightforward to see that this tree is truly a
junction tree. But in this tree $X$ and $Y$ are adjacent, so we can add $Z$
between them.

To prove the other direction, assume that adding $Z$ is a legal merge.  By
definition it immediately implies that there are sets $V, W \in \rf{\ifam{F};
S}$ such that $x \in V$ and $y \in W$. We need to show that $V \neq W$.  Let
$\ifam{H}$ be the family resulted from the merge. Let $\tree{U}$ be a junction
tree for $\ifam{H}$. Let $X \in \max\fpr{\ifam{F}}$ be an itemset such that $S
\cup \set{x} \subseteq X$ and, similarly, let $Y \in \max\fpr{\ifam{F}}$ such
that $S \cup \set{y} \subseteq Y$. We may assume that $X \neq S \cup \set{x}$
and $Y \neq S \cup \set{y}$, otherwise the proof is trivial. Hence, $X, Y \in
\max\fpr{\ifam{H}}$
Let $P$ be a path in $\tree{U}$ from $X$ to $Y$.  Let $R$ be a path in
$\tree{U}$ from $Z$ to $X$ and let $P_e$ be the first entry in path $R$ that
also occurs in $P$. We must have $Z \cap X \subseteq P_e$. The path from $Z$ to
$Y$ must also contain $P_e$, and so $Z \cap Y \subseteq P_e$. This implies that
$Z = (Z \cap X) \cup (Z \cap Y) \subseteq P_e$.  Since $Z$ is a maximal set, we
must have $Z = P_e$.

Itemset $Z$ is the only maximal itemset in $\ifam{H}$ that contains $x$ and $y$
simultaneously. Otherwise, Theorem~\ref{thr:legalsplit} states that $\splt{Z,
x, y}$ is illegal in $\ifam{H}$ so $\merge{S, x, y}$ is illegal in
$\ifam{F}$.

We can remove $Z = P_e$ from $\tree{U}$, attach $P_{e - 1}$ and $P_{e + 1}$ to
each other, and connect the rest adjacent nodes either to $P_{e - 1}$ or to
$P_{e + 1}$ depending whether these nodes have $x$ or $y$ as a member.  The
outcome, say $\tree{R}$, is a junction tree for $\ifam{F}$.

If $V = W$, then there must be a sequence 
\[
	X = \enpr{T_1}{T_N} = Y \in \ifam{F}
\]
such that $S \subsetneq T_i \cap T_{i + 1}$. Now consider a path $O$
in $\tree{R}$ visiting each $T_i$ in turn. Since $\tree{R}$ is a junction tree
we must have $S \subsetneq O_i \cap O_{i + 1}$. Since $X = T_1$ and $Y = T_N$
path $O$ must use the edge $(P_{e - 1}, P_{e + 1})$. But we must have $P_{e -
1} \cap P_{e + 1} = S$, otherwise $\tree{U}$ would not have been a junction
tree. This contradiction implies that $V \neq W$.
\end{proof}

\begin{proof}[of Theorem~\ref{thr:ergodic}]
We will prove the theorem by showing that all models can be brought by legal
splits to the indepenence model. Since a split $\splt{X, x, y}$ and a merge
$\merge{X - \set{x, y}, x, y}$, are the opposite operations we can reach model
$M$ from $M'$ by first reaching the independence model with splits and then
reaching $M$ from the independence model by merges.

Let $\tree{T}$ be a junction tree for $\ifam{F}' = \fm{M'}$ and let
$X \in \max\fpr{\ifam{F}'}$ be a maximal itemset 
such that $X$ is the leaf clique in a junction tree $\tree{T}$ with
$\abs{X} \geq 2$.  If no such $X$ exist, then $M'$ is in fact the independence
model.  Let $Y$ be the maximal itemset into which $X$ is connected in $\tree{T}$.

There must be an item $x \in X$ such that $x$ is not contained in any other
maximal itemset of $\ifam{F}'$. To prove this, assume otherwise. Then for each
item $z$ there is a maximal itemset $Z$ that contains $z$. The path from $X$ to
$X$ must go through $Y$ so the running intersection property implies that $x
\in Y$. But this implies that $X \subseteq Y$ which contradicts the maximality
of $X$. Let $y \in X$ such that $x \neq y$. Theorem~\ref{thr:legalsplit} now
states that $\splt{X, x, y}$ is a legal operation. We can now iteratively apply
this step until we reach the independence model. This proves the theorem.
\end{proof}

\begin{proof}[of Theorem~\ref{thr:ratio}]
For notational convinience, let us write
\[
	a(M) = \sum_{(X, Y) \atop \in E(\tree{T})} \ent{X \cap Y} \text{ and } b(M) = \sum_{X \in \max\fpr{\ifam{F}}} \ent{X}.
\]
Theorem~\ref{thr:decomposelikelihood} and Eq.~\ref{eq:bic} implies that
\[
	\log A = \abs{D}\fpr{a(M) - b(M)} - \frac{\log \abs{D}(\abs{\ifam{F}} - 1)}{2}
\]
and
\[
	\log B = \abs{D}\fpr{a(M') - b(M')} - \frac{\log \abs{D}(\abs{\ifam{F}'} - 1)}{2}.
\]
Let us write $d = \abs{D}(a(M') - b(M') - a(M) + b(M))$.  We will first show that $d =
\gain{X, x, y}$.  Let us denote $S = X - \set{x, y}$, $U = X - x$, and $V = X -
y$.  Consider the possibility that $U$ is not maximal in $\ifam{F}'$, in such
case there must be a maximal itemset $Q \in \ifam{F}$ such that $U \subset Q$
and that $Q$ is adjacent to $X$ with a separator $U$.  Similarly, if $V$ is not
maximal in $\ifam{F}'$, then there is a maximal itemset $P \in \ifam{F}$ such
that $V \subset P$ and that $P$ is adjacent to $X$ with a separator $V$.  By
studying the first part of the proof of Theorem~\ref{thr:legalsplit} we see
that there are four possible cases depending whether $U$ and/or $V$ is maximal
set in $\ifam{F}'$. The split operations are 

\begin{align*}
\begin{minipage}[b]{2cm}\hfill\begin{tikzpicture}[>=latex',line join=bevel,scale=0.23]\pgfsetlinewidth{0.5bp}\input{pics/split1_before}\end{tikzpicture}\end{minipage} & \begin{minipage}{0.7cm}\centering$\longrightarrow$\vspace{0.8em}\end{minipage}  \begin{minipage}[b]{2cm}\begin{tikzpicture}[>=latex',line join=bevel,scale=0.23]\pgfsetlinewidth{0.5bp}\input{pics/split1_after}\end{tikzpicture}\end{minipage} \tag*{\raisebox{0.5em}{(Case 1)}}\\
\begin{minipage}[b]{2cm}\hfill\begin{tikzpicture}[>=latex',line join=bevel,scale=0.23]\pgfsetlinewidth{0.5bp}\input{pics/split2_before}\end{tikzpicture}\end{minipage} & \begin{minipage}{0.7cm}\centering$\longrightarrow$\vspace{0.8em}\end{minipage}  \begin{minipage}[b]{2cm}\begin{tikzpicture}[>=latex',line join=bevel,scale=0.23]\pgfsetlinewidth{0.5bp}\input{pics/split2_after}\end{tikzpicture}\end{minipage} \tag*{\raisebox{0.5em}{(Case 2)}}\\
\begin{minipage}[b]{2cm}\hfill\begin{tikzpicture}[>=latex',line join=bevel,scale=0.23]\pgfsetlinewidth{0.5bp}\input{pics/split3_before}\end{tikzpicture}\end{minipage} & \begin{minipage}{0.7cm}\centering$\longrightarrow$\vspace{0.8em}\end{minipage}  \begin{minipage}[b]{2cm}\begin{tikzpicture}[>=latex',line join=bevel,scale=0.23]\pgfsetlinewidth{0.5bp}\input{pics/split3_after}\end{tikzpicture}\end{minipage} \tag*{\raisebox{0.5em}{(Case 3)}}\\
\begin{minipage}[b]{3cm}\hfill\begin{tikzpicture}[>=latex',line join=bevel,scale=0.23]\pgfsetlinewidth{0.5bp}\input{pics/split4_before}\end{tikzpicture}\end{minipage} & \begin{minipage}{0.7cm}\centering$\longrightarrow$\vspace{0.8em}\end{minipage}  \begin{minipage}[b]{2cm}\begin{tikzpicture}[>=latex',line join=bevel,scale=0.23]\pgfsetlinewidth{0.5bp}\input{pics/split4_after}\end{tikzpicture}\end{minipage} \tag*{\raisebox{0.5em}{(Case 4)}}\\
\end{align*}

In all cases, the difference $d$ is equal to $\ent{X} + \ent{S} - \ent{V} -
\ent{U} = \gain{X, x, y}$.  For example, in Case 2, the term $\ent{X}$ is in
$b(M)$ but not in $b(M')$, similarly the term $\ent{U}$ is in $b(M')$ but not
in $b(M)$, consequently $b(M) - b(M') = \ent{X} - \ent{U}$.  Similarly, $a(M) -
a(M') = \ent{V} - \ent{S}$. Thus $d = \gain{X, x, y}$ for Case 2, and similar
observations show that $d = \gain{X, x, y}$ holds also for other cases.

To complete the proof we need to show that $\abs{\ifam{F}} - \abs{\ifam{F}'} =
2^{\abs{X} - 2}$.  During a split, itemsets containing $x$ and $y$ are removed,
but there are exactly $2^{\abs{X} - 2}$ such itemsets. We can now combine these
results
\[
\begin{split}
	\log B - \log A & = d - \frac{\log \abs{D}}{2}(\abs{\ifam{F}} - \abs{\ifam{F}'}) \\
	                & = \gain{X, x, y} - \log \abs{D}2^{\abs{X} - 3}, \\
\end{split}
\]
which proves the theorem.
\end{proof}

\begin{proof}[of Theorem~\ref{thr:reduceseparate}]
Let $X_i$ and $Y_i$ be the maximal itemsets such that $X_i \cap Y_i = S_i$,
$x_i \in X_i$ and $y_i \in Y_i$. These sets must exist since adding $Z_i$ is a
legal merge.

Assume that $Z_1 = Z_2$. This implies that either all items $x_1$, $y_1$,
$x_2$, and $y_2$ are unique or two of them are the same. Assume the latter case
and assume that $x_1 = x_2$. Then we must have $y_2 \in S_1 - S_2$ and $y_1 \in
S_2 - S_1$. We have $x_1 \in X_1 \cap X_2$, $y_1 \in X_2 \cap Y_1$, and $y_2
\in Y_1 \cap X_1$.  Hence the conditions in Lemma~\ref{lem:cycle} hold and
$\ifam{F}$ is not decomposable.

Assume now that all four items are unique. This implies that $x_1, y_1 \in S_2
- S_1$ and $y_2 \in S_1 - S_2$. Again we have $x_1 \in X_1 \cap X_2$, $y_1
\in X_2 \cap Y_1$, and $y_2 \in Y_1 \cap X_1$ so Lemma~\ref{lem:cycle} implies
that $\ifam{F}$ is not decomposable.
\end{proof}

\begin{proof}[of Theorem~\ref{thr:reducednumber}]
Let $S$ be an itemset. We will show that $\mdgr{M, S} > 0$ only if $S$ is a
separator, that is, $S = X \cap Y$ for two adjacent itemsets in a junction
tree.  The theorem will follow from this, since a junction tree contains at most
$K$ nodes and hence at most $K - 1$ edges.

To prove the result, assume that $\mdgr{M, S} > 0$. This implies that there are
two sets, say $U, V \in \rf{\ifam{F}, S}$. Let $x \in U$ and $y \in V$.  This
implies that there are (at least) two sets $X$ and $Y$ such that $S \cup
\set{x} \subset X$ and $S \cup \set{y} \subset Y$. Let $P = \enpr{P_1}{P_L}$ be
a path from $X$ to $Y$ in the junction tree. The running intersection property
implies that $S \subseteq P_i \cap P_{i + 1}$ If there is no $i$ such that $S =
P_i\cap P_{i + 1}$, then $X$ and $Y$ should have been joined together during
the construction of $\rf{\ifam{F}, S}$. In other words, there must be a set in
$\rf{\ifam{F}, S}$ containing $x$ and $y$. Hence, there are adjacent itemsets
$P_i$ and $P_{i + 1}$ such that $S = P_i\cap P_{i + 1}$. This proves the
theorem.
\end{proof}

\begin{proof}[of Theorem~\ref{thr:splitupdate}]
Adding $X = S \cup \set{x, y}$ into $\ifam{F}$ can only make splits illegal.
Assume that $\splt{Y, a, b}$ becomes illegal after adding $X$.
Theorem~\ref{thr:legalsplit} implies that $a, b \in X$.  We must have $\set{a,
b} \not\subseteq S$ since otherwise the split would not be legal in the
original family. Assume that $a = x$ and let $W \in \max\fpr{\ifam{F}}$ is the
maximal itemset containing $\set{x} \cup S$ (such itemset exists because of the definition of merge). Note that, $b \neq y$, so $b \in S$ and
$\set{a, b} \in W$.  If there is another maximal
itemset, say $Z \in \max\fpr{\ifam{F}}$, such that $\set{x, b} = \set{a, b} \subset Z$,
then the split $\splt{Y, x, b}$ is not legal in the original family. This
proves the theorem.
\end{proof}

\begin{proof}[of Theorem~\ref{thr:mergeupdate}]
Let $Z = S \cup \set{x, y}$. This itemset is the only maximal itemset in
$\ifam{G}$ that contains $x$ and $y$ simultaneously. The definition of reduced
family now implies immediately Claims 2 and 3. 

To prove Claim 1 let $V, W \in \rf{\ifam{F}; S}$ such that $x \in V$ and $y \in
W$. These must be separate sets because of Theorem~\ref{thr:reduceseparate}. The set $Z$ now
connects these sets into one in $\rf{\ifam{G}; S}$.

To see Claim 4 let $T$ be an itemset. If $x, y \in T$, then 
$\rf{\ifam{G}; T}$ has only one set and $\rf{ifam{G}; T}$ has zero sets.
If $x \in T$ and $y \notin T$, then $T - S \neq \emptyset$,
otherwise $T$ is covered by Claim 2 or 3. In such case, the construction of
$\rf{\ifam{G}; T}$ does not use $Z$, and hence remains unchanged. Assume that
$x, y \notin T$.  Let $X, Y \in \max\fpr{\ifam{F}}$ be the itemsets such that $S \cup
\set{x} \subseteq X$ and $S \cup \set{y} \subseteq Y$. If $T - S \neq
\emptyset$, then $Z$ is not used. Assume now that $S \subsetneq T$. In
this case the items of $X - T$  and $Y - T$ are already joined into one
itemset, hence $\rf{\ifam{G}; T}$ remains unchanged.
\end{proof}

%% file: pics/split1_before.tex
\pgfsetcolor{black}
\begin{scope}
  \definecolor{strokecol}{rgb}{0.0,0.0,0.0};
  \pgfsetstrokecolor{strokecol}
  \draw (27bp,27bp) ellipse (27bp and 27bp);
  \draw (27bp,27bp) node {$X$};
\end{scope}

%% file: pics/split1_after.tex
\pgfsetcolor{black}
  \draw [] (55.6bp,27bp) .. controls (74.962bp,27bp) and (100.72bp,27bp)  .. (119.96bp,27bp);
  \definecolor{strokecol}{rgb}{0.0,0.0,0.0};
  \pgfsetstrokecolor{strokecol}
  \draw (88bp,45.5bp) node {$S$};
\begin{scope}
  \definecolor{strokecol}{rgb}{0.0,0.0,0.0};
  \pgfsetstrokecolor{strokecol}
  \draw (28bp,27bp) ellipse (27bp and 28bp);
  \draw (28bp,27bp) node {$U$};
\end{scope}
\begin{scope}
  \definecolor{strokecol}{rgb}{0.0,0.0,0.0};
  \pgfsetstrokecolor{strokecol}
  \draw (147bp,27bp) ellipse (27bp and 27bp);
  \draw (147bp,27bp) node {$V$};
\end{scope}

%% file: pics/split2_before.tex
\pgfsetcolor{black}
  \draw [] (54.136bp,27bp) .. controls (73.443bp,27bp) and (99.22bp,27bp)  .. (118.22bp,27bp);
  \definecolor{strokecol}{rgb}{0.0,0.0,0.0};
  \pgfsetstrokecolor{strokecol}
  \draw (86bp,45.5bp) node {$V$};
\begin{scope}
  \definecolor{strokecol}{rgb}{0.0,0.0,0.0};
  \pgfsetstrokecolor{strokecol}
  \draw (27bp,27bp) ellipse (27bp and 27bp);
  \draw (27bp,27bp) node {$X$};
\end{scope}
\begin{scope}
  \definecolor{strokecol}{rgb}{0.0,0.0,0.0};
  \pgfsetstrokecolor{strokecol}
  \draw (144bp,27bp) ellipse (25bp and 25bp);
  \draw (144bp,27bp) node {$P$};
\end{scope}

%% file: pics/split2_after.tex
\pgfsetcolor{black}
  \draw [] (55.666bp,27bp) .. controls (75.148bp,27bp) and (101.05bp,27bp)  .. (120.12bp,27bp);
  \definecolor{strokecol}{rgb}{0.0,0.0,0.0};
  \pgfsetstrokecolor{strokecol}
  \draw (88bp,45.5bp) node {$S$};
\begin{scope}
  \definecolor{strokecol}{rgb}{0.0,0.0,0.0};
  \pgfsetstrokecolor{strokecol}
  \draw (28bp,27bp) ellipse (27bp and 28bp);
  \draw (28bp,27bp) node {$U$};
\end{scope}
\begin{scope}
  \definecolor{strokecol}{rgb}{0.0,0.0,0.0};
  \pgfsetstrokecolor{strokecol}
  \draw (146bp,27bp) ellipse (25bp and 25bp);
  \draw (146bp,27bp) node {$P$};
\end{scope}

%% file: pics/split3_before.tex
\pgfsetcolor{black}
  \draw [] (55.76bp,27bp) .. controls (75.575bp,27bp) and (102.1bp,27bp)  .. (121.79bp,27bp);
  \definecolor{strokecol}{rgb}{0.0,0.0,0.0};
  \pgfsetstrokecolor{strokecol}
  \draw (89bp,45.5bp) node {$U$};
\begin{scope}
  \definecolor{strokecol}{rgb}{0.0,0.0,0.0};
  \pgfsetstrokecolor{strokecol}
  \draw (28bp,27bp) ellipse (27bp and 28bp);
  \draw (28bp,27bp) node {$Q$};
\end{scope}
\begin{scope}
  \definecolor{strokecol}{rgb}{0.0,0.0,0.0};
  \pgfsetstrokecolor{strokecol}
  \draw (149bp,27bp) ellipse (27bp and 27bp);
  \draw (149bp,27bp) node {$X$};
\end{scope}

%% file: pics/split3_after.tex
\pgfsetcolor{black}
  \draw [] (55.6bp,27bp) .. controls (74.962bp,27bp) and (100.72bp,27bp)  .. (119.96bp,27bp);
  \definecolor{strokecol}{rgb}{0.0,0.0,0.0};
  \pgfsetstrokecolor{strokecol}
  \draw (88bp,45.5bp) node {$S$};
\begin{scope}
  \definecolor{strokecol}{rgb}{0.0,0.0,0.0};
  \pgfsetstrokecolor{strokecol}
  \draw (28bp,27bp) ellipse (27bp and 28bp);
  \draw (28bp,27bp) node {$Q$};
\end{scope}
\begin{scope}
  \definecolor{strokecol}{rgb}{0.0,0.0,0.0};
  \pgfsetstrokecolor{strokecol}
  \draw (147bp,27bp) ellipse (27bp and 27bp);
  \draw (147bp,27bp) node {$V$};
\end{scope}

%% file: pics/split4_before.tex
\pgfsetcolor{black}
  \draw [] (55.76bp,27bp) .. controls (75.575bp,27bp) and (102.1bp,27bp)  .. (121.79bp,27bp);
  \definecolor{strokecol}{rgb}{0.0,0.0,0.0};
  \pgfsetstrokecolor{strokecol}
  \draw (89bp,45.5bp) node {$U$};
  \draw [] (176.14bp,27bp) .. controls (195.44bp,27bp) and (221.22bp,27bp)  .. (240.22bp,27bp);
  \draw (208bp,45.5bp) node {$V$};
\begin{scope}
  \definecolor{strokecol}{rgb}{0.0,0.0,0.0};
  \pgfsetstrokecolor{strokecol}
  \draw (28bp,27bp) ellipse (27bp and 28bp);
  \draw (28bp,27bp) node {$Q$};
\end{scope}
\begin{scope}
  \definecolor{strokecol}{rgb}{0.0,0.0,0.0};
  \pgfsetstrokecolor{strokecol}
  \draw (266bp,27bp) ellipse (25bp and 25bp);
  \draw (266bp,27bp) node {$P$};
\end{scope}
\begin{scope}
  \definecolor{strokecol}{rgb}{0.0,0.0,0.0};
  \pgfsetstrokecolor{strokecol}
  \draw (149bp,27bp) ellipse (27bp and 27bp);
  \draw (149bp,27bp) node {$X$};
\end{scope}

%% file: pics/split4_after.tex
\pgfsetcolor{black}
  \draw [] (55.666bp,27bp) .. controls (75.148bp,27bp) and (101.05bp,27bp)  .. (120.12bp,27bp);
  \definecolor{strokecol}{rgb}{0.0,0.0,0.0};
  \pgfsetstrokecolor{strokecol}
  \draw (88bp,45.5bp) node {$S$};
\begin{scope}
  \definecolor{strokecol}{rgb}{0.0,0.0,0.0};
  \pgfsetstrokecolor{strokecol}
  \draw (28bp,27bp) ellipse (27bp and 28bp);
  \draw (28bp,27bp) node {$Q$};
\end{scope}
\begin{scope}
  \definecolor{strokecol}{rgb}{0.0,0.0,0.0};
  \pgfsetstrokecolor{strokecol}
  \draw (146bp,27bp) ellipse (25bp and 25bp);
  \draw (146bp,27bp) node {$P$};
\end{scope}